\documentclass[letterpaper]{article} % DO NOT CHANGE THIS
\usepackage{aaai25}  % DO NOT CHANGE THIS
\usepackage{times}  % DO NOT CHANGE THIS
\usepackage{helvet}  % DO NOT CHANGE THIS
\usepackage{courier}  % DO NOT CHANGE THIS
\usepackage[hyphens]{url}  % DO NOT CHANGE THIS
\usepackage{graphicx} % DO NOT CHANGE THIS
\usepackage{natbib}  % DO NOT CHANGE THIS AND DO NOT ADD ANY OPTIONS TO IT
\usepackage{caption} % DO NOT CHANGE THIS AND DO NOT ADD ANY OPTIONS TO IT
\usepackage{graphicx}
\usepackage{amssymb}
\usepackage{algorithm}
\usepackage{algorithmicx}
\usepackage{subcaption} 
\usepackage{algpseudocode}
\usepackage{makecell}
\usepackage{multirow}
\usepackage{amsmath}
\usepackage{amsthm}
\usepackage{newfloat}
\usepackage{listings}
\DeclareCaptionStyle{ruled}{labelfont=normalfont,labelsep=colon,strut=off} % DO NOT CHANGE THIS
\floatstyle{ruled}
\newfloat{listing}{tb}{lst}{}
\floatname{listing}{Listing}
\title{EBS-CFL: Efficient and Byzantine-robust Secure Clustered Federated Learning}
\author{
    Zhiqiang Li\textsuperscript{\rm 1}, Haiyong Bao\textsuperscript{\rm 1}\thanks{Corresponding author.}, Menghong Guan\textsuperscript{\rm 1}, Hao Pan\textsuperscript{\rm 1}, Cheng Huang\textsuperscript{\rm 2}, Hong-Ning Dai\textsuperscript{\rm 3}
}
\affiliations{
    \textsuperscript{\rm 1} East China Normal University\\
    \textsuperscript{\rm 2} Fudan University\\
    \textsuperscript{\rm 3} Hong Kong Baptist University\\
    \{51255902170,mhguan,51255902119\}@stu.ecnu.edu.cn, hybao@sei.ecnu.edu.cn, \\
    chuang@fudan.edu.cn, hndai@ieee.org
}

\usepackage{bibentry}
\begin{document}

\maketitle
\newtheorem{definition}{Definition} 
\newtheorem{theorem}{Theorem}
\newtheorem{assumption}{Assumption}
\begin{abstract}
Despite {\it federated learning} (FL)'s potential in collaborative learning, its performance has deteriorated due to the data heterogeneity of distributed users. Recently, {\it clustered federated learning} (CFL) has emerged to address this challenge by partitioning users into clusters according to their similarity. However, CFL faces difficulties in training when users are unwilling to share their cluster identities due to privacy concerns. 
To address these issues, we present an innovative Efficient and Robust Secure Aggregation scheme for CFL, dubbed \textbf{EBS-CFL}. 
The proposed EBS-CFL supports effectively training CFL while maintaining users' cluster identity confidentially. Moreover, it detects potential poisonous attacks without compromising individual client gradients by discarding negatively correlated gradients and aggregating positively correlated ones using a weighted approach. The server also authenticates correct gradient encoding by clients. 
EBS-CFL has high efficiency with client-side overhead $\mathcal{O}(ml + m^2)$ for communication and $\mathcal{O}(m^2l)$ for computation, where $m$ is the number of cluster identities, and $l$ is the gradient size. When $m = 1$, EBS-CFL's computational efficiency of client is at least $\mathcal{O}(\log{n})$ times better than comparison schemes, where $n$ is the number of clients.
In addition, we validate the scheme through extensive experiments.
Finally, we theoretically prove the scheme's security.
\end{abstract}

\begin{links}
    \link{Code}{https://github.com/Lee-VA/EBS-CFL}
    % \link{Datasets}{https://aaai.org/example/datasets}
    % \link{Extended version}{https://aaai.org/example/extended-version}
\end{links}

\section{Introduction}
Federated learning (FL) \cite{FL_McMahan}, as a distributed machine learning paradigm, decentralizes model training across multiple local devices, thereby preserving data privacy. However, traditional FL methods operate under the assumption of independently and identically distributed (i.i.d.) data, which limits their effectiveness in real-world scenarios with non-i.i.d. data. To address this challenge, Clustered Federated Learning (CFL) \cite{CFL_propose} has been proposed. CFL enhances model performance on heterogeneous datasets by grouping participants into clusters based on data similarity, effectively extending FL's applicability to non-i.i.d. scenarios.

Most CFL methods, such as recent researches \cite{IFCA, FedMDS, EDSI}, mainly focus on improving accuracy or performance efficiency. For example, {\it Iterative Federated Clustering Algorithm (IFCA)} \cite{IFCA} is the first to provide convergence guarantees, which demonstrates its effectiveness in non-convex problems of neural networks. 
Despite numerous studies on CFLs, they still face challenges incurred by Byzantine attacks and privacy leakage. Particularly, they do not consider protecting against clients' cluster identities, which could potentially leak clients' important privacy.

As highlighted in DLG \cite{DLG}, the vulnerability of FL is exposed when gradient inference attacks target unencrypted gradients.
Though Differential Privacy (DP) methods provide some protection, their noise can hinder model convergence. The integration of DP with secure aggregation in researches \cite{fpsa, pbm} lessens DP's security requirements and tackles differential attacks, but efficiency is still problematic.
Secure aggregation methods \cite{PSA, CCESA, FastSecAgg, SwiftAgg_plus} can protect plaintext gradients, yet they fall short in dealing with Byzantine attacks. Robust algorithms, as achieved in PEFL \cite{PEFL}, rely on interaction between the server and the cloud platform, but this approach assumes no collusion between the two, which is a relatively weak security assumption. Lastly, SAFELearning \cite{SAFELearning} operates under the assumption of a single cloud server and supports only the FedAvg algorithm, overlooking compatibility with other federated learning algorithms.
The current research primarily faces two core challenges:
\begin{table*}[ht]
\small
\centering
\begin{tabular}{ccccc}
        \hline  
        \textbf{Approach} & \makecell[c]{\textbf{Malicious} \\ \textbf{Adversary}} & \makecell[c]{\textbf{Defend Against} \\ \textbf{Byzantine Attacks}} & \textbf{Single-Server} & \makecell[c]{\textbf{Compatible} \\ \textbf{Framework}} \\  
        \hline  
        CCESA \cite{CCESA} & $\times$ & $\times$ & $\checkmark$ & FedAvg \\
        FastSecAgg \cite{FastSecAgg} & $\times$ & $\times$ & $\checkmark$ & FedAvg \\  
        SwiftAgg+ \cite{SwiftAgg_plus} & $\times$ & $\times$ & $\checkmark$ & FedAvg \\ 
        PEFL \cite{PEFL} & $\checkmark$ & $\checkmark$ & $\times$ & FedAvg \\ 
        SAFELearning \cite{SAFELearning} & $\checkmark$ & $\checkmark$ & $\checkmark$ & FedAvg \\
        EBS-CFL & $\checkmark$ & $\checkmark$ & $\checkmark$ & FedAvg, IFCA\\ 
        \hline  
    \end{tabular}  
\caption{Comparing related secure federated learning schemes.}
\label{tab:compare}
\end{table*}

\begin{table*}[t]
\small
    \begin{tabular}{ccccc}
      \hline  
      ~ & \multicolumn{2}{c}{\textbf{Communication}}  &  \multicolumn{2}{c}{\textbf{Computation}}\\
      \textbf{Approach} & \textbf{Server} & \textbf{Per-client}  &  \textbf{Server} & \textbf{Per-client}\\  
      \hline  
      SecAgg & $\mathcal{O}(nl+sn^2)$ & $\mathcal{O}(l+sn)$ & $\mathcal{O}(n^2l)$ & $\mathcal{O}(nl+sn^2)$ \\  
      CCESA & $\mathcal{O}(nl+sn\sqrt{n\log{n}})$ & $\mathcal{O}(l+s\sqrt{n\log{n}})$ & $\mathcal{O}(nl\log{n})$ & $\mathcal{O}(l\sqrt{n\log{n}} + sn\log{n})$ \\  
      FastSecAgg & $\mathcal{O}(nl+n^2)$ & $\mathcal{O}(l+n)$ & $\mathcal{O}(l\log{n})$ & $\mathcal{O}(l\log{n})$ \\  
      EBS-CFL (m=1) & $\mathcal{O}(nl)$ & $\mathcal{O}(l)$ & $\mathcal{O}(n^2 + nl\log{n})$ & $\mathcal{O}(l)$ \\
      EBS-CFL & $\mathcal{O}(nml+n)$ & $\mathcal{O}(ml+m^2)$ & $\mathcal{O}(m^2(nml + nm^2 + n^2 + nl\log{n}))$ & $\mathcal{O}(m^2l)$ \\
      \hline  
    \end{tabular} 
\caption{Communication and computation overhead. For better clarity, we have included the case with the number of clusters $ m = 1 $ in the table, comparing it with other schemes that do not incorporate the concept of multiple clusters.}
\label{tab:communication_computation}
\end{table*}

\textbf{Lack of a scheme to protect clients' cluster identities.} Existing solutions require the acquisition of clients' cluster identities for training. However, there is no privacy preservation scheme for cluster identities. In this way, attackers may use the clients' cluster identities as additional information to launch inference attacks. 

\textbf{Lack of a comprehensive secure aggregation solution.} Despite the advent of secure aggregation, the comprehensive solution of ensuring both efficiency and Byzantine-robustness in a single cloud server setting with compatibility of diverse FL schemes is still largely elusive.

To address these challenges, we propose a novel {\it Efficient and Byzantine-robust Secure Clustered Federated Learning} (EBS-CFL) scheme.
In particular, we devise a {\it Robust Federated Clustering Algorithm (RFCA)} with the integration of a Byzantine robust algorithm based on cosine similarity \cite{FLTrust} into the aggregation process of IFCA \cite{IFCA}. Additionally, we protect clients' gradients and cluster identities by developing a secure aggregation scheme. This scheme allows the server to perform clustered aggregation without accessing the clients' gradients or clusters' identities while filtering out malicious gradients. We also introduce compression scheme to guarantee the scalability of communication and computation costs. 
Particularly in terms of communication overhead, our scheme is more advantageous compared to other schemes. When the number of clusters is $1$, the server's overhead is linearly related to the number of clients, while the overhead for an individual client is independent of the total number.
Our main contributions are summarized as follows.
\begin{itemize}
    \item Our scheme can protect the privacy of client gradients and cluster identities when training models. Specifically, we designed the Verifiable Orthogonal Matrix Confusion for Aggregation (VOMCA), which enables clients to encode both their cluster identities and gradient information simultaneously. This innovative design allows for the aggregation of gradients according to set cluster identities, without exposing clients' identities.
    \item We designed a secure aggregation scheme to achieve high efficiency, Byzantine-robustness, a single cloud server assumption, and compatibility with CFL schemes. 
    We use the assumption of a single secure cloud server, which is more secure than those based on multiple cloud servers.
    Specifically, we designed RFCA, which realizes Byzantine-robust CFL. Meanwhile, we designed the Secure ReLU Function Computation Mechanism (SRFC) to filter out malicious gradients in secure aggregation. Lastly, we designed a compression scheme to significantly enhance the overall efficiency of the system.

    \item We conduct theoretical analysis and extensive experiments to evaluate our proposed EBS-CFL. Through theoretical analysis, we explain the communication and computational complexity of our scheme. Moreover, we conducted experiments involving multiple variables related to communication and computation overhead, conducting detailed data analysis of the entire process of aggregation stage. And we incorporate experiments involving typical Byzantine attacks in FL to validate the robustness of our proposed EBS-CFL. Finally, we demonstrate the security of the EBS-CFL through security analysis.
\end{itemize}

Table~\ref{tab:compare} compares EBS-CFL with other state-of-the-art works in terms of key characteristics. 
Table~\ref{tab:communication_computation} compares EBS-CFL with other state-of-the-art works in terms of communication and computational complexity. 
Note that for detailed illustration of theoretical analysis of complexity, please refer to the appendix.

\begin{figure}[ht] 
    \centering  
    \includegraphics[width=1\columnwidth, clip]{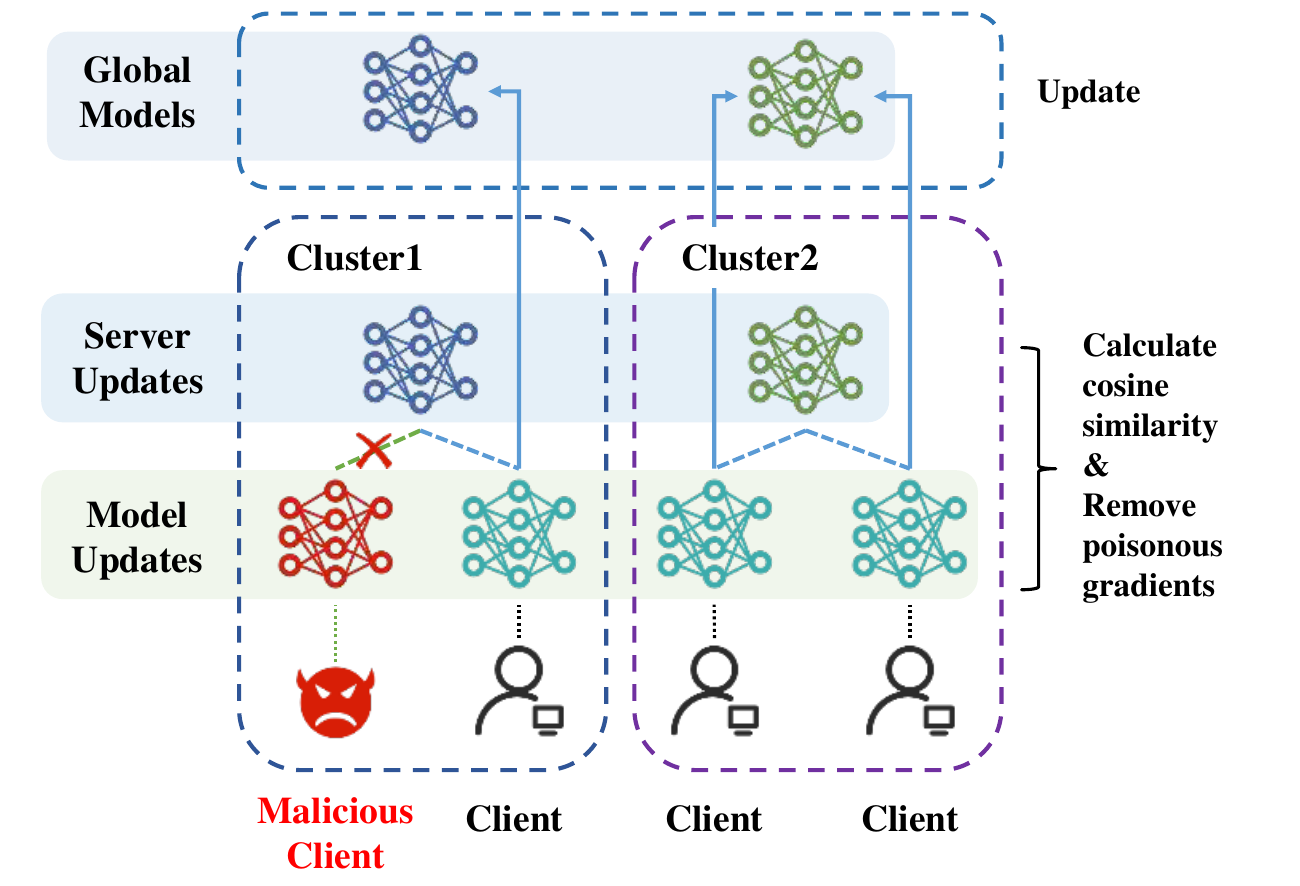}  
    \caption{Overview of RFCA.}  
    \label{fig:RFCA}
\end{figure}

\section{Related Work}
We provide an overview of relevant research from two perspectives: \textbf{efficient secure aggregation} and \textbf{defense against Byzantine attacks}.

In the realm of \textbf{efficient secure aggregation}, PSA \cite{PSA} introduced a lightweight protocol that allows the server to aggregate model updates from multiple clients without accessing individual updates. CCESA \cite{CCESA} presented a low-complexity scheme using sparse random graphs to ensure data privacy through a secret-sharing node topology. FastSecAgg \cite{FastSecAgg} developed a protocol, which leverages the Fast Fourier Transform (FFT) to create FastShare, a novel multi-secret sharing scheme that enables more efficient computation and communication. SwiftAgg+ \cite{SwiftAgg_plus} proposed a secure aggregation protocol that significantly reduces communication overhead while preserving privacy, achieving optimal communication loads and security guarantees. However, these schemes do not address Byzantine attacks.

In the realm of \textbf{defense against Byzantine attacks}, PEFL \cite{PEFL} proposed a framework that uses homomorphic encryption and penalizes poisoning behavior by extracting gradient data with logarithmic functions, addressing the problems of both privacy leakage and poisonous attack. SAFELearning \cite{SAFELearning} introduced technique that supports secure aggregation through Unintentional Random Grouping (ORG) and Partial Parameter Disclosure (PPD), effectively preventing malicious attacks. ShieldingFL \cite{ShieldingFL} proposed an adaptive client selection strategy to combat Byzantine attacks by filtering out honest clients, but it does not encrypt or obfuscate gradients.

All existing solutions build upon the traditional FL framework, which struggles with handling non-i.i.d. data. As the demand for handling non-i.i.d. data grows, the recently proposed CFL has provided a solution. Therefore, there is an urgent need for a secure aggregation framework that can be compatible with CFL.

\section{System Model}
Our system involves three entities: the \textbf{Clients}, the \textbf{Server}, and the \textbf{Key Distribution Center (KDC)}. The functions of each entity are as follows.

\begin{itemize}
    \item \textbf{Clients}: Clients encode gradients obtained from local training using encoding keys and upload the encoded gradients to the server. Malicious clients may attempt to bypass the robustness algorithm by altering encoded gradients or generating them incorrectly, thereby impacting the global model.
    \item \textbf{Server}: The server broadcasts the model, collects client-uploaded gradients for validation, and performs weighted clustering aggregation using our proposed EBS-CFL. Malicious servers may improperly execute the robustness algorithm to obtain target gradients.
    \item \textbf{Key Distribution Center (KDC)}: The KDC generates and distributes keys.
\end{itemize}

\section{Proposed Scheme}
This section introduces our proposed EBS-CFL. First, we present a Byzantine-robust clustered federated learning algorithm, RFCA. Based on this, we designed a secure aggregation scheme for privacy preservation. Finally, we improved the overall efficiency of the system through an efficient compression scheme.

\subsection{Robust Federated Clustering Algorithm (RFCA)}
We first design Algorithm~\ref{alg:cmu} to train clustered gradient, and then design Robust Federated Clustering Learning Algorithm (RFCA).
The overview of the RFCA scheme is illustrated Figure~\ref{fig:RFCA}.
\begin{algorithm}[tp]
\caption{ClusteredModelUpdate($\{\theta_i\}_{i=1}^m$, $D$, $b$, $\eta$, $T$)}  
\label{alg:cmu}
\begin{algorithmic}[1]
\Statex \textbf{Input:} initialization $\{\theta_i\}_{j=1}^m$; training datasets $D$; batch size $b$; learning rate $\eta$; number of iterations $T$. 
    \State $\{\theta^0_i\}_{i=1}^m \gets \{\theta_i\}_{i=1}^m$  
    \State $Loss(D; \theta^0_j) \gets \min(\{Loss(D_b; \theta^0_i\}_{i=1}^m))$ 
    \For{$t=1$ to $T$}
        \State Randomly sample a batch $D_b$ from $D$.
        \State $\theta^t_j \gets \theta^{t-1}_j - \eta \nabla Loss(D_b; \theta^{t-1}_j)$ 
    \EndFor  
    \State \Return $j, \theta^T_j - \theta^0_j$
\end{algorithmic}  
\end{algorithm}

The specific steps of RFCA are as follows.
\begin{itemize}
    \item The server sets $m$ clusters and initializes $m$ models. Then the server collects a small and clean public dataset and trains $m$ models to obtain $m$ server updates. Finally, $m$ server updates are obtained, represented as $\{g_0^j\}_{j=1}^m$.
    \item The clients run Algorithm~\ref{alg:cmu} to train local model.
    \item The server calculates cosine similarities between server updates and clients' gradients, and removes poisonous gradients. Then, for $1 \le i,j \le n$ and $1 \le k \le m$, the server calculates
    \begin{equation}
    \label{eq:rfca_aggregation}
        g^k = \frac{1}{\sum_{j=1}^{n} \text{ReLU}(c_j^k)} \sum_{i=1}^{n} \text{ReLU}(c_i^k) \frac{\lVert g_0^k \rVert}{\lVert g_i \rVert} \cdot g_i,
    \end{equation}
    where $c_i^k = s_{i,k}\frac{\langle g_i, g_0 \rangle}{\lVert g_i \rVert \cdot \lVert g_0 \rVert}$, $k$ represents the cluster identity, $\{s_{i,k}\}_{k=1}^m$ represents a one-hot encoding vector ($s_{i,k} = 1$, when $k$ matches the target identity), and $g_i$ represents the gradient submitted by the $i$-th client.
\end{itemize}

Note that for detailed illustration of the pseudo code of RFCA, please refer to the appendix.

\subsection{Verifiable Orthogonal Matrix Confusion for Aggregation (VOMCA)} 
VOMCA is a privacy-preserving aggregation technology. Similar to secret sharing, individual data is only decryptable after aggregation. It also ensures that clients cannot bypass robustness checks and submit malicious gradients.

\begin{definition}
    Given a set of matrices $\{M_i\}_{i=1}^l$, they are called mutually orthogonal matrices \cite{mo} if they satisfy the following condition: 
    \begin{equation}
        M_i^T M_j = 
        \begin{cases}
            I, &  i = j,\\ 
            O, &  i \neq j, 
        \end{cases}
    \end{equation}
    where $I$ and $O$ respectively represent an identity matrix and a zero matrix, for $1 \leq i, j \leq l$.
\end{definition}

For $1 \le i \le n$, the $i$-th client can encode the secret $x_i$ as follows.
\begin{equation}
    \chi_i = x_iM_i + R_iM_{i+n},
\end{equation}
where $\sum_{i=1}^n R_i = O$ and $\{M_i\}_{i=1}^{2n}$ are Mutually Orthogonal Matrices. Construct verification key $vk$ as $\{\sum_1^n \mathcal{V}_iM_{i+n}, \{\mathcal{V}_i \cdot R_i^T\}_{i=1}^n\}$. The verification function is as follows.
\begin{equation}
\label{eq:verify}
        \mathcal{V}(\chi_i, vk) =
        \begin{cases}
            1, & vk \cdot \chi_i^T = \mathcal{V}_i \cdot R_i^T, \\
            0, & vk \cdot \chi_i^T \ne \mathcal{V}_i \cdot R_i^T,
        \end{cases}
    \end{equation}
where $\mathcal{V}_i$ is a random matrix that has the same shape with $R_i$. And  construct decode key $dk$ as $\sum_1^{2n} M_i$, which satisfies the following relations, 
\begin{equation*}
    \chi_i \cdot dk^T = x_i+R_i, \quad \text{for} \; 1 \le i \le n,
\end{equation*}
\begin{equation*}
    (\sum_{i=1}^n \chi_i) \cdot dk^T = \sum_{i=1}^n x_i.
\end{equation*}
\begin{theorem}
\label{th:verify}
    For any $1 \le i \le n$, $a_i \in \{\chi_i, \chi_i'\}$, where $\chi_i'$ represents a ciphertext tampered by an adversary. The $\mathcal{V}$ satisfies
    \begin{equation}
    \label{eq:p_verify}
        Pr[\sum_{i=1}^k a_i dk^T \ne \sum_{i=1}^k x_i \wedge \forall i \in [n], \mathcal{V}(\chi_i, vk)=1] \le \epsilon,
    \end{equation}

where $\epsilon = \frac{1}{\lambda^l}$, $\lambda$ represents the probability of correctly predicting pseudo-random numbers without additional information, and $l$ represents the size of $x_i$.
\end{theorem}
\begin{proof}
    For detailed illustration of the proof process, please refer to the appendix.
\end{proof}
\subsection{Secure ReLU Function Computation Mechanism (SRFC)}
We designed SRFC to exclude gradients with negative cosine similarity, while ensuring security.

We define the function $\mathcal{F}$ to square each element in the matrix, preserving their original signs. Meanwhile, we define $\mathcal{F}^{-1}$ to calculate the square root of each element's absolute value, also preserving their original signs. The definitions of $\mathcal{F}$ and $\mathcal{F}^{-1}$ are as follows.
\begin{equation*}
    \mathcal{F}(A) = (\frac{(A)_{ij}^3}{|(A)_{ij}|}), \;
    \mathcal{F}^{-1}(A) = (\frac{(A)_{ij}\sqrt{|(A)_{ij}|}}{|(A)_{ij}|}),
\end{equation*}
where $A \in \mathbb{R}^{a \times b}$. 
For $1 \le i \le n$, input $x_i \in \mathbb{R}$, define the encode of $x_i$ as
\begin{equation*}
    E(x_i)=x_iM_i + \frac{1}{2}\zeta_iM_{i+n} + R_i.
\end{equation*}
Construct $\{\alpha_i\}_{i=1}^n$ as follows.
\begin{equation*}
    \alpha_i = x_i \cdot M_i^TR_i'M_i + M_{i+n}^TR_i'M_{i+n}, 
\end{equation*}
where $\forall M \in \{M_i\}_{i=1}^{2n}, M \in \mathbb{R}^{lm \times 2lmn}$, and $\{M_i\}_{i=1}^{2n}$ are Mutually Orthogonal Matrices.
And construct $\beta$, $\tau_1$, $\tau_2$, and $\tau_3$ as follows.
\begin{equation}
\small
\label{eq:ReLU_params}
    \left\{
    \begin{aligned}
        \beta =& \sum_{i=1}^{2n} M_i, \\
        \tau_1 =& \sum_{i=1}^n M_i^T{R_i'}^{-1}M_i 
         + M_{i+n}^T{R_i'}^{-1}(R_i^{(0)} + \zeta_i \cdot M_{i+n}), \\
        \tau_2 =& \sum_{i=1}^n M_i^T{R_i'}^{-2}(M_i \circ M_i) \\
        &+ M_{i+n}^T{R_i'}^{-2}(R_i^{(1)} \circ R_i^{(1)} + R_i^{(2)}), \\
        \tau_3 =& \sum_{i=1}^n M_i^T{R_i'}^{-2}\mathcal{F}(2 \cdot M_i \circ R_i^{(1)}) 
        - M_{i+n}^T{R_i'}^{-2}\mathcal{F}(R_i^{(2)}),
    \end{aligned}
    \right.
\end{equation}

where $\circ$ denotes hadamard product, $\zeta_i \in \mathbb{R}^{lm \times lm}$, $\sum_{i=1}^n \zeta_i = 0$, $R_i' \in \mathbb{R}^{lm \times lm}$, $R_i'$ is an invertible matrix, $R_i, R_i^{(0)}, R_i^{(1)}, R_i^{(2)} \in \mathbb{R}^{lm \times 2lmn}$, $R_i=R_i^{(0)} + R_i^{(1)}$, $\sum_{i=1}^n R_i = 0$. 

For $1 \le j \le lm$, $1 \le k \le 2lmn$, $r^{(1)}_{jk}$ and $r^{(2)}_{jk}$ are random numbers, the $R_i^{(1)}$ and $R_i^{(2)}$ satisfy that
if $(M_i)_{jk} \le 0 $,
\begin{equation}
\label{eq:random_sample_matrix1}
    \text{Pr}[(R_i^{(1)})_{jk} = r^{(1)}_{jk} \wedge (R_i^{(2)})_{jk} = 0] = 1,
\end{equation}
where $0 < r_{jk} < \max(\{|x_i|\}_{i=1}^n)$; if $(M_i)_{jk} > 0$,
\begin{equation}
\label{eq:random_sample_matrix2}
    \text{Pr}[(R_i^{(1)})_{jk} = 0 \wedge (R_i^{(2)})_{jk} = r^{(2)}_{jk}] = 1,
\end{equation}
where $0 < r^{(2)}_{jk} < 2 \cdot \max(\{|x_i|\}_{i=1}^n)^2$. 

The principle of SRFC is as follows. Firstly, by squaring the matrix elements and then taking the square root, we can obtain the absolute value of the corresponding elements. Then, by adding this absolute value to the original element, we can obtain the output of the ReLU function. 

The $\alpha$ is the encoding of $x$, and when it performs matrix multiplication operation with itself, the decoding can yield $x^2$. The function of $\beta$ is to perform decoding, while $\tau_1$, $\tau_2$, and $\tau_3$ help to compute the square of the target element and then compute the square root. Since the entire computation process requires privacy preservation, a random number factor is introduced and incorporated into all parameters.

The $\sum_{i=1}^n E(ReLU(x))$ and $\sum_{i=1}^n ReLU(x_i)$ can be calculated as theorem \ref{th:srfc}.
\begin{theorem}
\label{th:srfc}
    The $\sum_{i=1}^n E(ReLU(x))$ can be calculated as
    \begin{equation}
        \label{eq:ReLU_encode}
        \begin{aligned}
            &\sum_{i=1}^n E(ReLU(x_i)) = \frac{1}{2}(\beta \cdot (\sum_{i=1}^n \alpha_i) \cdot \tau_1 \\
            &+ \sum_{i=1}^n \mathcal{F}^{-1}(\beta \cdot \alpha_i^2 \cdot \tau_2 + \mathcal{F}^{-1}(\beta \cdot \alpha_i^2 \cdot \tau_3))),
        \end{aligned}
    \end{equation}
    and the $\sum_{i=1}^n ReLU(x_i)$ can be calculated as
    \begin{equation}
        \label{eq:ReLU_decode}
        \sum_{i=1}^n ReLU(x_i) = \frac{1}{l}tr(\sum_{i=1}^n E(ReLU(x_i) \cdot \beta^T).
    \end{equation}
\end{theorem}
\begin{proof}
    For detailed illustration of the proof process, please refer to the appendix.
    % \ref{pf:srfc}.
\end{proof} 
\subsection{Secure Aggregation Scheme}
Based on RFCA, we design our secure aggregation scheme. In each round of training, the number of participating clients is denoted by $n$, the number of cluster identities is denoted by $m$, and the dimension of the gradients is denoted by $l$.

\textbf{Initialization:} KDC pre-generates multiple sets of Mutually Orthogonal Matrices $\{A_i | A_i \in \mathbb{R}^{l \times lm}\}_{i=1}^m$, $\{M_i | M_i \in \mathbb{R}^{lm \times 2lmn}\}_{i=1}^{2n}$, and $\{M_i' | M_i' \in \mathbb{R}^{lm \times 3lmn}\}_{i=1}^{3n}$. The server sends the update gradient $\{g_0^i\}_{j=1}^m$ to KDC. Then KDC generates $\beta$, $\tau_1$, $\tau_2$, and $\tau_3$ as Eq.\eqref{eq:ReLU_params}. And KDC generates the encode keys $\{ek_i\}_{i=1}^n$, $\{\alpha'_i\}_{i=1}^{2lmn}$, $vk$, and $dk$ as follows.

The $(M)_i$ represents the $i$-th row of matrix $M$, for $1 \le i \le n$ and $1 \le j \le m$, 
\begin{equation}
\label{eq:ek}
     ek_i = \{ek_i^0, ek_i^1\}= \{M_i', \sum_{j=1}^m R''_{ij}A_jM'_{i+n} + \mu_iM'_{i+2n}\},
\end{equation}
where $\sum_{i=1}^n R''_{ij} = 0$, $||R''_{ij}||=1$, $\sum_{i=1}^n \mu_i = 0$, and $||\mu_i||=1$.
And for $1 \le i \le 2lmn$,
\begin{equation}
\label{eq:alpha_dot}
     \begin{aligned}
        \alpha'_i =& \sum_{j=1}^n {M_j'}^T (\sum_{k=1}^m A_k^T(g_0^k)^T) (M_j^TR_j'
        M_j)_i \\
        +& {M_{j+2n}'}^T {\mu_j}^{-1} (M_{j+n}^TR_j'M_{j+n})_i. 
    \end{aligned}
\end{equation}
\begin{equation}
    \left\{
    \begin{aligned}
        vk =& \{\sum_{i=1}^n \mathcal{V}_i\sum_{j=1}^m A_j M'_{i+n},  \{\mathcal{V}_i \cdot \sum_{j=1}^m {R''_{ij}}^T \}_{i=1}^n\}, \\
        dk =& \sum_{i=1}^n M_i^TM'_i + M_{i+n}^T\sum_{j=1}^n (M'_{j+n} + M'_{j+2n}),
    \end{aligned}
    \right.
\end{equation}
where $\{\mathcal{V}_i\}_{i=1}^n$ are random matrices generated by KDC.

\textbf{Broadcast Phase:} 
In the first round, the KDC sends $ek_i$ and $\{A_i\}_{i=1}^m$ to the $i$-th client. For subsequent rounds, the KDC updates $ek_i^1$ by updating $R''_{ij}$ and sends the updated $ek_i^1$ to the $i$-th client. The server first classifies the current models using the public dataset, then sends the global model to the selected clients.

\textbf{Training Phase:} 
The training phase follows the same procedure as RFCA. 
% Assuming the cluster of the $i$-th client is $j$, the gradient is encoded as $\delta_i = \frac{g_i}{\lVert g_i \rVert}A_jM'_i + \sum_{j=1}^m R''_{ij}A_jM'_{i+n} + \mu_iM'_{i+2n}$.
The cluster corresponding to the $i$-th client is denoted as $j$, the gradient is encoded as $\delta_i = \frac{g_i}{\lVert g_i \rVert}A_jM'_i + \sum_{j=1}^m R''_{ij}A_jM'_{i+n} + \mu_iM'_{i+2n}$. The encoded gradient $\delta_i$ is then uploaded.

\textbf{Aggregation Phase:}
The KDC sends $\{\alpha_i\}_{i=1}^{lm}$, $\beta$, $\tau_1$, $\tau_2$, $dk$, and $vk$ to the server. The server then verifies each client as described in Eq.\eqref{eq:verify}, and checks if $\lVert \delta_i \rVert^2 = 3$ for $1 \le i \le n$. 
If the verification in Eq.\eqref{eq:verify} passes, it can be conclusively demonstrated that the client is unable to manipulate the ciphertext. If $\lVert \delta_i \rVert^2 = 3$, it can be inferred that the client has performed a normalization operation $\frac{g_i}{\lVert g_i \rVert}$.
If the verification fails, the server will request the failed clients to resend their encoded gradients, or the server will restart the process. For $1 \le i \le n$, the server calculates $\alpha_i$ as follows,

\begin{equation}
    (\alpha_i)_j = \delta_i \cdot \alpha'_j.
\end{equation}

The server calculates $\sum_{i=1}^n E(ReLU(w_i))$ and $\sum_{i=1}^n ReLU(w_i)$ as Eq.\eqref{eq:ReLU_encode} and Eq.\eqref{eq:ReLU_decode}. 

Finally, the server calculates the aggregated gradient $g$ as follows,
\begin{equation}
    dk' = \frac{\sum_{i=1}^n E(ReLU(w_i)) \cdot dk}{\sum_{i=1}^n ReLU(w_i)},
\end{equation}
\begin{equation}
    g = (\sum_{i=1}^n \delta_i) \cdot {dk'}^T.
\end{equation}

\textbf{Update Phase:} For $1 \le j \le m$, the server updates as follows.
\begin{equation}
    \theta^{(t)}_j = \theta^{(t-1)}_j - \eta \cdot ||g_0^j|| \cdot gA_j^T.
\end{equation}
\subsection{Secure Key Transformation Mechanism (SKT)}
For $1 \le i \le n$, $1 \le j \le m$, and converse function $\mathcal{C}: i \rightarrow j$, construct Transformation key $tk_i$ as $ M_i^TM_{\mathcal{C}(i)}' + M_{i+n}^T(M_{\mathcal{C}(i+n)}' + R_i^{-1} \cdot R_i')$ for the key transformation of $\chi_i$, and the Transformation function is as follows.
\begin{equation}
     \mathcal{T}(\chi_i, tk_i) = \chi_i \cdot tk_i = \chi_i' + R_i',
\end{equation}
where $c$ represents a constant, $\{M_i\}_{i=1}^{n}$ and $\{M_j'\}_{j=1}^{m}$ are Mutually Orthogonal Matrices, $\chi_i = x_iM_i + R_iM_{i+n}$, and $\chi'_i = x_iM_{\mathcal{C}(i)}' + R_iM_{\mathcal{C}(i+n)}'$. It satisfies 
\begin{equation}
\label{eq:transfer_deocde}
    (\sum_{i=1}^m \chi_i') \cdot (dk')^T = \sum_{i=1}^n c \cdot x_i,
\end{equation}
where $dk' = \sum_{i=1}^{m}M_i'$. 

\subsection{Compression Scheme}

We design gradient segmentation to reduce gradient dimensionality for individual aggregations and layered aggregation to minimize the number of clients involved in each aggregation.

\textbf{Gradient Segmentation:}
The $i$-th client flattens the local gradients into a vector by rows, and concatenates them layer by layer to form a vector $g_i$. Let the vector $g_i$ be divided into $s$ segments, denoted as $\{g_{ij}\}_{j=1}^s$, each segment has a length of $t$, which satisfies that $s \cdot t \ge l$. For $1 \le i \le n$, $1 \le j \le m$, and $1 \le k \le s$, the $R''$ and $\mu$ satisfy
\begin{equation*}
    \sum_{k=1}^s ||R_{ijk}''||^2 = 1, \;
    ||\mu_{ijk}||^2 = \frac{1}{s}.
\end{equation*}
For $1 \le i \le n$, $\mathcal{A}(i)$ represents the $i$-th client, the corresponding inner products are calculated as follows.
\begin{equation}
\label{eq:gs_innerP}
    g_i{g_0^{\mathcal{A}(i)}}^T = \sum_{j=1}^s g_{ij}{g_{0,j}^{\mathcal{A}(i)}}^T.
\end{equation}
For $1 \le i \le 3mnt$, $1 \le j \le n$, and $1 \le k \le s$, following the gradient segmentation, Eq.\eqref{eq:alpha_dot} should be reformulated as follows.
\begin{equation*}
     \begin{aligned}
        \alpha'_{ik} =& \sum_{j=1}^n {M_j'}^T (\sum_{k=1}^m A_k^T(g_0^k)^T) (M_j^T
        M_j)_i \\
        +& \frac{1}{s}{M_{j+2n}'}^T {\mu_{jk}}^{-1} (M_{j+n}^TR_j'M_{j+n})_i. 
    \end{aligned}
\end{equation*}
Finally, aggregate the corresponding segments and concatenate the segments back to the original gradients.

\textbf{Layered Aggregation:}
The total number of clients is denoted as $n$, and the clients are divided into $\xi$ groups,  each of which is with a maximum size of $\lceil \frac{n}{\xi} \rceil$.

KDC generates $ek$, $vk$, $\alpha_i'$, and transformation keys for each group. Then, KDC generates $\beta$, $\tau_1$, $\tau_2$, and $\tau_3$. In Addition, $\zeta$, $\mu$, $R''$, and $\{A_i\}_{i=1}^m$ are shared across all clients.
The specific steps for aggregation are as follows.

\begin{itemize}
    \item First, the steps before aggregation for each group are the same as the secure aggregation scheme. In the aggregation phase, the server will calculate $\sum_{i=1}^n E(ReLU(w_i))$ as Eq.\eqref{eq:ReLU_encode}. However, since each group cannot eliminate the random matrices and lacks a corresponding $dk$ prior to key transformation, the decoding is not performed.
    \item Second, the server performs keys transformation in each group.
    \item Finally, the server decodes the encoded gradients and subsequently updates the global model.
\end{itemize}
Each transformation can be regarded as the aggregation of each group. The number of transformation rounds can be denoted as $x$. The total size $\mathcal{S}(n, \{\xi_i\}_{i=1}^x)$ of transformation keys can be expressed as follows.
\begin{equation*}
    \mathcal{S}(n, \{\xi_i\}_{i=1}^x) = 2lm\xi_x + \sum_{i=1}^{x-1} 4l^2m^2\xi_i\xi_{i+1}\prod_{j=i+1}^x \xi_j,
\end{equation*}
where the $\{\xi_i\}_{i=1}^x$ satisfies that $\prod_{i=1}^x \xi_i \ge n$. 
For $1 \le i \le n$, $1 \le j \le x$, the key transformation is as follow.
\begin{equation}
\label{eq:la}
    \delta_i^j = \mathcal{T}(\delta_i^{j-1}, \sum_{k=1+\nu}^{\xi_j+\nu} tk_k^j),
\end{equation}
where $\nu = i-(i\mod{\xi_j})$, and $\delta_i^0 = \delta_i$. The $\alpha_i$ has the same transformation process as Eq.\eqref{eq:la}.

\section{Convergence Guarantees}
In this section, we provide convergence guarantees for RFCA. Our analysis builds upon the assumptions of IFCA, with additional considerations specific to the Byzantine-robustness algorithm.

\begin{theorem}
\label{th:convergence}
Let all assumptions hold, and the step size $\gamma$ be chosen as $\gamma = \frac{1}{L}$. Suppose that at a certain iteration of the RFCA algorithm, the parameter vector $\theta_j$ satisfies:  
$\|\theta_j - \theta^*_j\| \leq \left(\frac{1}{2} - \alpha\right) \frac{\lambda}{L} \Delta,$
where $0 < \alpha < \frac{1}{2}$. Let $\theta^+_j$ denote the next iterate, $m$ is the total number of clients, $k$ is the number of clusters, and $w = \sum_{i=1}^m w_i$. For any fixed $j \in [k]$ and $\delta \in (0, 1)$, the following holds with probability at least $1 - \delta$,
\begin{equation}
    \begin{aligned}
        \|\theta^+_j - \theta^*_j\| &\leq \left(1 - \frac{p\lambda}{8L}\right) \|\theta_j - \theta^*_j\| + \frac{c_0\sqrt{2v^2 + 1}}{\delta L\sqrt{pmn'}} \\
        &+ c_1\frac{\eta^2m}{\delta \alpha^2\lambda^2\Delta^4 wn'} + c_2\frac{\eta k\sqrt{2kmv^2 + km}}{\delta^{\frac{3}{2}}\alpha\lambda L\Delta^2wn'}.
    \end{aligned}
\end{equation}
\end{theorem}
This theorem guarantees the convergence of RFCA. Note that for details on assumptions, parameter definitions, theoretical derivations, and the proof, please refer to the appendix.

\section{Performance Evaluation}
In this section, we will evaluate the performance of our proposed EBS-CFL, including both efficiency and Byzantine-robustness. We carry out experiments on MNIST \cite{deng2012mnist}, CIFAR10, and CIFAR100 \cite{Krizhevsky_2009}. Note that more detailed experiments will be presented in the appendix.

\subsection{Efficiency}
In this section, we will evaluate the efficiency of our scheme by contrasting it with solutions that do not incorporate privacy preservation, thereby demonstrating the optimization we have achieved in our privacy-preserving scheme.

We evaluate the average client communication, as shown in Figures~\ref{fig:communication1} and \ref{fig:communication2}. Experiments showed that the communication overhead does not increase with the number of clients, and it has a linear relationship with the size of the data vector. Under normal values of $m$, the communication overhead does not significantly increase with changes in other variables, proving the practicality of the method.

We conducted similar experiments for server-side computational overhead, with the number of cluster identities set to 2, as shown in Figure~\ref{fig:computation}. We measured the efficiency by calculating the consumed time. Experiments showed that the time consumption of our scheme is almost equivalent to FedAvg, corresponding to the theoretical analysis of computational overhead.

\begin{figure}[ht] 
    \centering  
    \includegraphics[width=0.5\textwidth,clip]{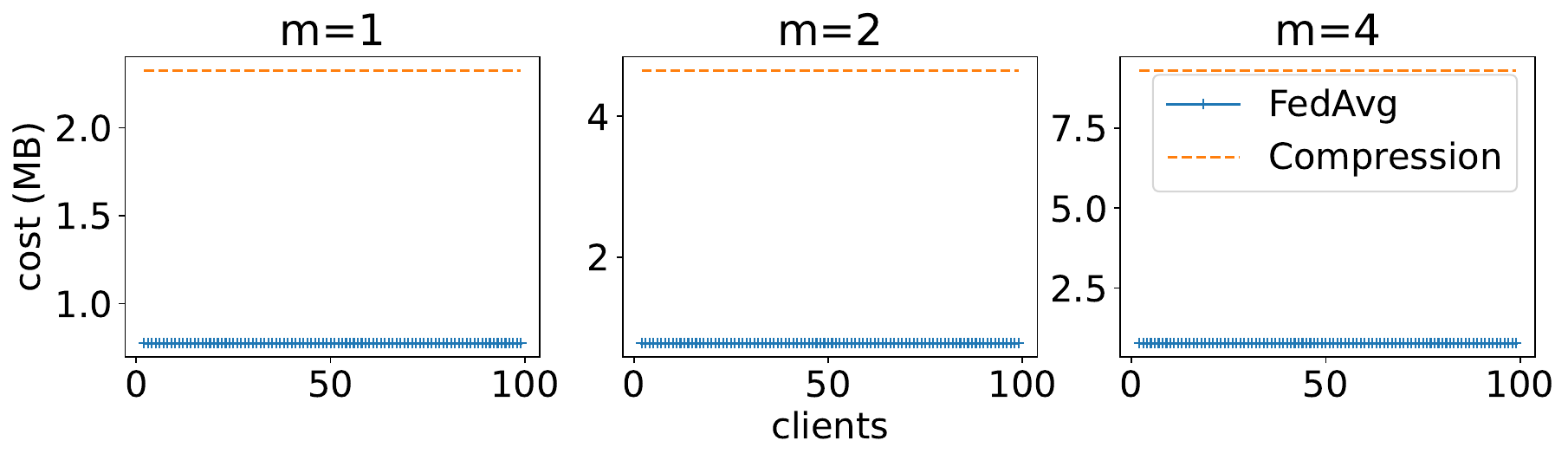}  
    \caption{With a fixed model size of 794KB, the impact of other factors on individual client communication overhead.}
  \label{fig:communication1} 
\end{figure}

\begin{figure}[ht] 
    \centering  
    \includegraphics[width=0.5\textwidth, clip]{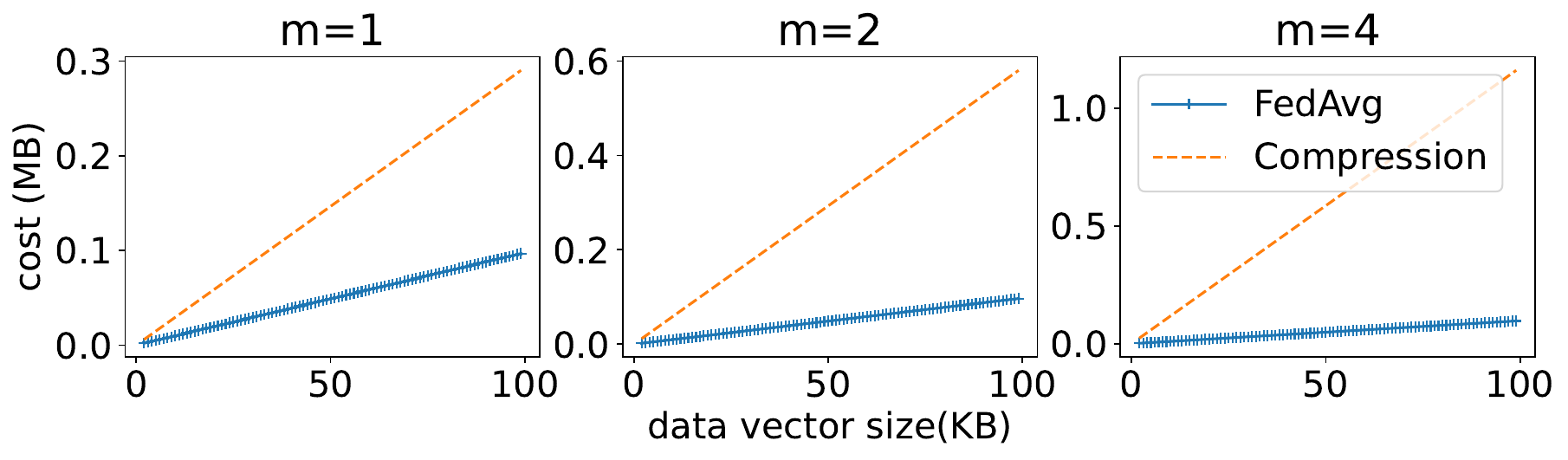}  
    \caption{With a fixed number of clients at 20, the impact of other factors on individual client communication overhead.}
  \label{fig:communication2}  
\end{figure}

\begin{figure}[h]  
  \centering  
  \subfloat[n=5]
    {\includegraphics[width=0.24\textwidth]{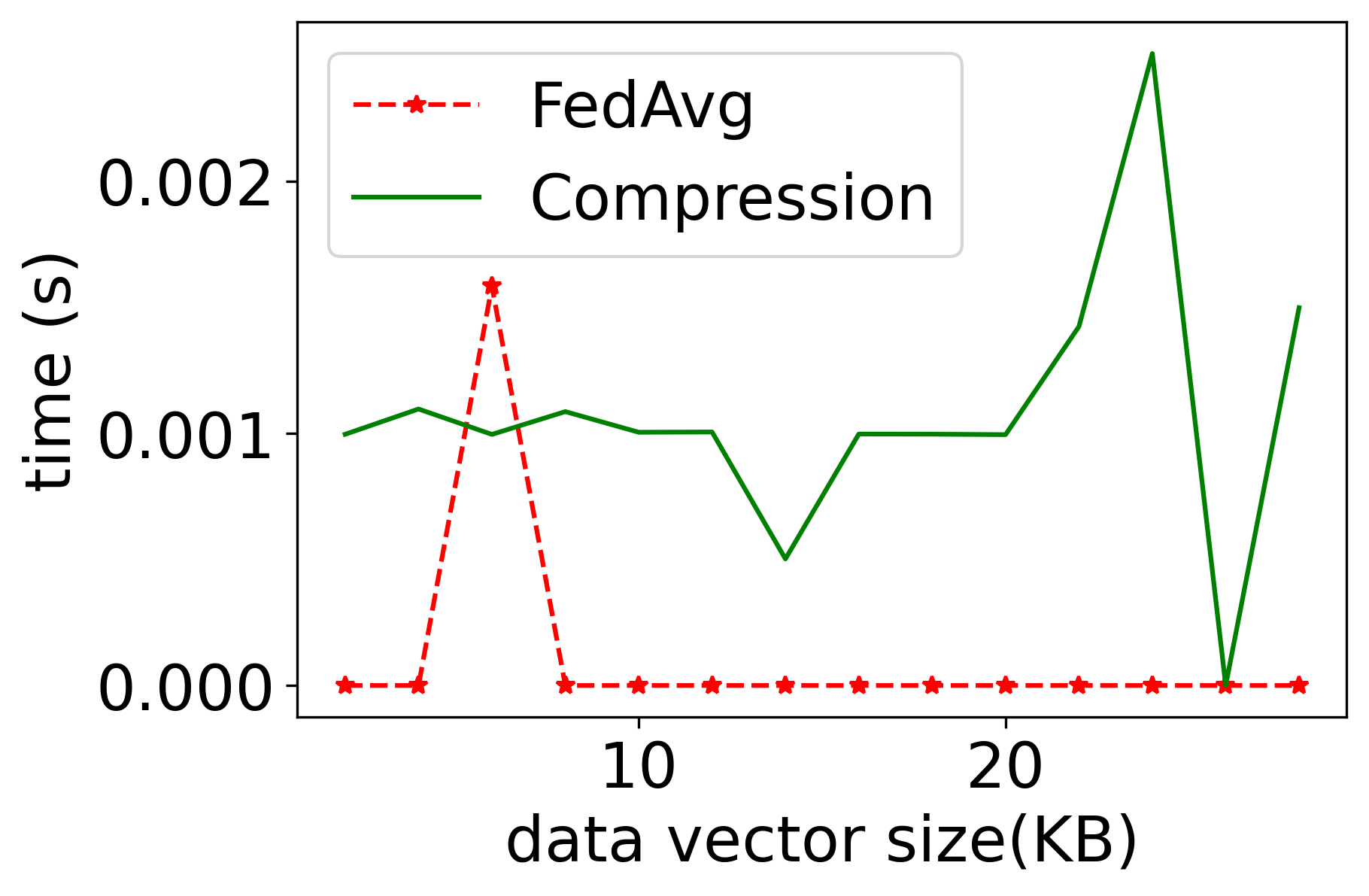}}
    \subfloat[l=10 KB]
    {\includegraphics[width=0.24\textwidth]{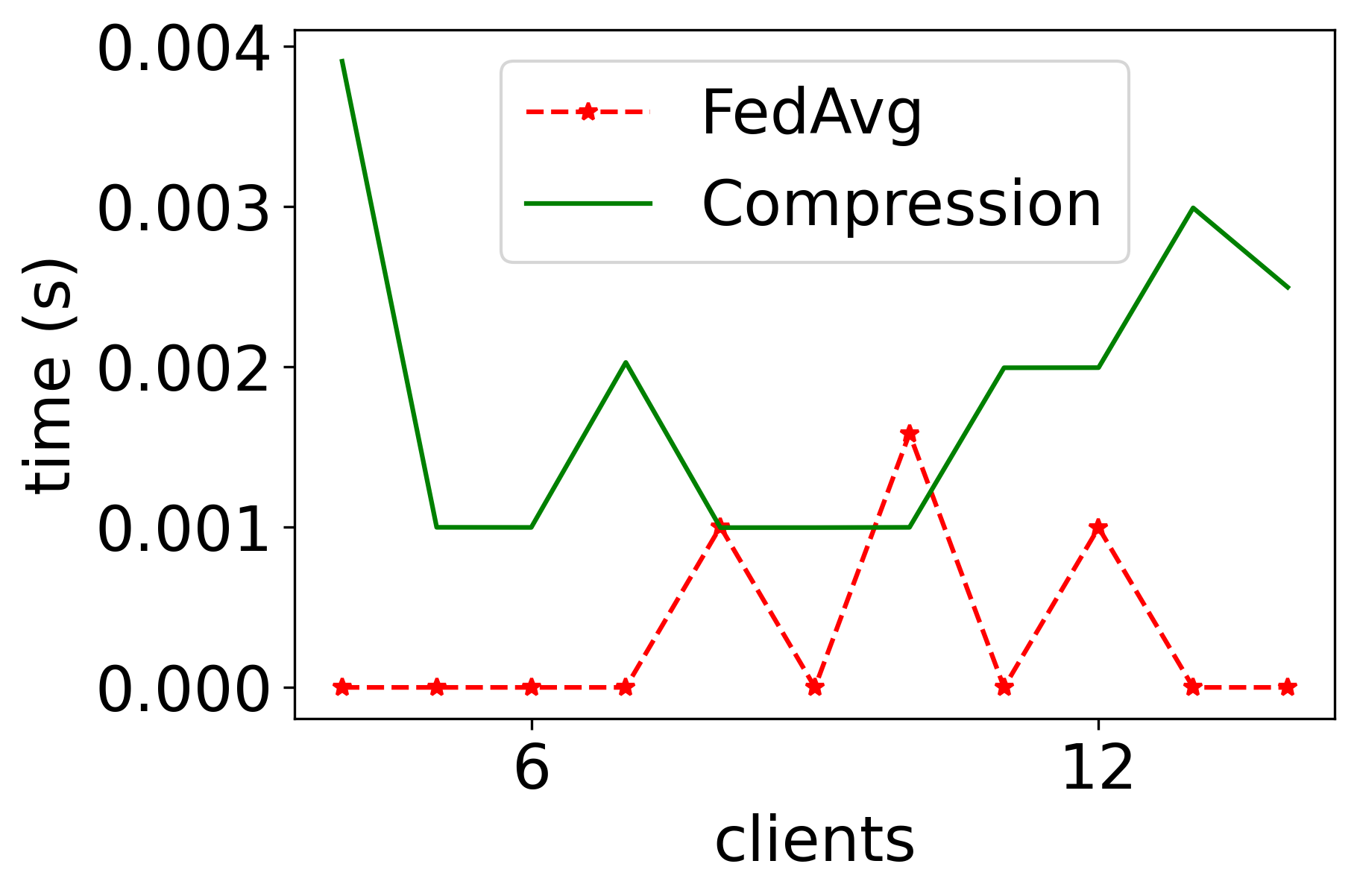}}
  \caption{Study the impact of other variables on computational overhead while fixing the number of clients and model size.}
  \label{fig:computation}  
\end{figure}  

\begin{table*}[ht]
\small
\centering
    \begin{tabular}{ccccc}
      \hline  
      ~ & ~ &  \multicolumn{3}{c}{\textbf{Data Heterogeneity}}\\
      \textbf{Approach} & \textbf{Attack} & \textbf{$\alpha$ = 0.1} & \textbf{$\alpha$ = 0.5} & \textbf{$\alpha$ = 0.9}\\  
      \hline  
      \multirow{2}{*}{FedAvg} & \textbf{LF} & 38.82 / 39.53 / 100.00 / 15.29 & 45.36 / 48.40 / 100.00 / 14.77 & 41.77 / 43.25 / 100.00 / 18.29 \\  
      ~ & \textbf{Sine}  & 12.78 / 40.22 / 100.00 / 26.52 & 14.28 / 43.71 / 100.00 / 28.21 & 43.61 / 48.86 / 100.00 / 15.51 \\
      \hline
      \multirow{2}{*}{FLtrust} & \textbf{LF} & 49.49 / 49.49 / 49.50 / 6.15 & 55.07 / 55.20 / 38.00 / 8.56 & 57.88 / 57.88 / 35.50 / 6.33 \\  
      ~ & \textbf{Sine} & 47.28 / 47.28 / 100.00 / 8.11 & 57.95 / 57.95 / 100.00 / 6.45 & 59.59 / 59.59 / 100.00 / 5.55 \\
      \hline
      \multirow{2}{*}{IFCA} & \textbf{LF} & 48.18 / 49.01 / 100.00 / 15.83 & 37.57 / 42.50 / 100.00 / 20.04 & 49.63 / 52.03 / 100.00 / 13.05 \\  
      ~ & \textbf{Sine} & 12.31 / 25.48 / 100.00 / 36.72 & 10.29 / 20.49 / 100.00 / 38.48 & 41.52 / 42.93 / 100.00 / 19.31 \\
      \hline
      \multirow{2}{*}{EBS-CFL} & \textbf{LF} & 62.13 / 64.58 / 7.00 / 1.73 & 65.29 / 66.30 / 0.00 / 0.25 & 65.88 / 65.90 / 0.50 / 2.10 \\  
      ~ & \textbf{Sine} & 60.56 / 61.95 / 9.00 / 3.38 & 65.45 / 66.17 / 31.50 / 0.74 & 66.74 / 67.09 / 37.00 / 1.85 \\
    \hline  
    \end{tabular} 
\caption{Comparing the performance of different approaches under Byzantine attacks in a non-i.i.d. setting. Each cell in the table includes the final accuracy (FA), maximum accuracy during training (MA), Attack Success Rate (ASR = $\frac{\text{successful attackers}}{\text{total attackers}}$), and Attack Impact Rate (AIR = $\frac{2 \cdot \text{NA} - \text{FA} - \text{MA}}{2 \cdot \text{NA}}$), separated by ``/''. NA represents the accuracy achieved during training without any attacks. For the sake of convenience in presentation, the percentage symbols in the table are omitted.}
\label{tab:data_heterogeneous_byzantine_attacks}
\end{table*}
\subsection{Byzantine-robustness}
In this section, we evaluate the Byzantine robustness of our scheme. To demonstrate the effectiveness of the proposed EBS-CFL, we compare it against FLTrust, a state-of-the-art Byzantine-robustness algorithm. For attacks, we consider the representative Label-Flipping (LF) Attack \cite{lf_attack} and the advanced Sine Attack \cite{Sine}. Additional comparisons with other methods, such as those in \cite{Median, Krum, PoisoningAttacks}, are detailed in the appendix.

To simulate data heterogeneity under a non-i.i.d. setting, we use Dirichlet's distribution \cite{MNIID}, where $\alpha$ controls data dispersion (smaller $\alpha$ indicates higher heterogeneity). We evaluate both attacks at a high adversarial rate of 40\%. The Sine Attack leverages cosine similarity to generate gradients that align closely with the central benign gradient but diverge from other client gradients. Our scheme protects server updates, rendering the AK-BSU assumption of Sine ineffective. Therefore, we conduct experiments under the realistic AK assumption and PC capability defined in Sine Attack. Results are shown in Table~\ref{tab:data_heterogeneous_byzantine_attacks}.
Greater data heterogeneity reduces training accuracy and increases attack success rates. However, since heterogeneity inherently affects training accuracy, the relative impact of attackers does not rise significantly. Malicious gradients introduced by attackers are classified as low-quality and assigned minimal weight, reducing their influence. 
FLTrust exhibits weakened defense under heterogeneous data. Against Sine Attack, FLTrust can be fully compromised as data heterogeneity increases, making it harder to distinguish malicious gradients from legitimate but highly divergent ones. LF Attack follows a similar trend but achieves a lower success rate than Sine Attack, despite causing a greater overall impact.
In contrast, our scheme significantly improves defense by effectively clustering and identifying malicious gradients in heterogeneous data. While accuracy declines when $\alpha = 0.1$, this is primarily due to the extreme heterogeneity rather than interference from Byzantine adversaries.

\section{Security Analysis}

In the security proof, we define the perturbation functions $\phi = x \cdot M$ and $\psi(x, R) = x + R$, which safeguard the privacy of matrix and vector operations, respectively. 
We also define the function $\varphi(\mathcal{S}) = \sum_{M \in \mathcal{S}} M$, where $\mathcal{S} \subset \{M_i\}_{i=1}^n$ and $\{M_i\}_{i=1}^n$ are mutually orthogonal matrices. The sum of these orthogonal matrices does not reveal information about individual matrices.

Using this framework, via theorems we prove that specific encryption mechanisms, such as VOMCA and SRFC, ensure the confidentiality and immutability of individual variables due to the incorporation of randomness. In secure aggregation and compression schemes, multiple randomizations and key transformations protect gradients and cluster identities from being leaked. Additionally, the SKT mechanism guarantees that neither keys nor variables can be inferred during transformations.
Note that for detailed illustration of the proof process, please refer to the appendix.

\section{Conclusion}
In this paper, we have proposed an efficient and robust clustered federated learning secure aggregation framework, EBS-CFL. Specifically, we construct an encoding mechanism using matrix techniques, allowing the server to filter out poisonous gradients without knowing the clients' gradients and cluster identities. We then perform weighted aggregation based on the correlation between the gradients and the server updates. In addition, we provide detailed theoretical proofs of the correctness and the security of the approach, and analyze its efficiency. Finally, through extensive experiments, we demonstrate the efficiency, effectiveness, and robustness of our approach.
% 在本文中，我们提出了一个高效具有鲁棒性的聚类联邦学习安全聚合框架 EBS-CFL。具体来说，我们通过矩阵技术构造了一种编码方法，使得服务器可以在不知道客户端的梯度和聚类身份的情况下，过滤有毒梯度，并根据梯度和服务器更新的相关性进行加权聚合。我们通过详细的理论证明了方案的正确性和安全性，并对效率进行了分析。最后通过大量实验，证明了我们方案的高效性，有效性和鲁棒性。

\section*{Acknowledgements}
The authors would like to express their gratitude to Haiyong Bao for his insightful discussions and guidance. We also thank Cheng Huang and Hong-Ning Dai for their productive and stimulating conversations. Finally, our heartfelt thanks go to Menghong Guan and Hao Pan for their invaluable support and assistance.

The authors are very grateful as this work was supported in part by the National Natural Science Foundation of China (62072404); and in part by the Natural Science Foundation of Shanghai Municipality (23ZR1417700). 

\bibliography{aaai25}

\appendix
\section{Supplementary Experiments}
This section will provide detailed supplementary information to demonstrate the effectiveness of our proposed EBS-CFL.

\subsection{Model Accuracy}
To verify the compatibility of our approach with CFL, we conducted experiments to assess model convergence. We trained the model on two datasets, i.e., MNIST and CIFAR10, with various numbers of clients and recorded performance on the test set during each round. We simulated uneven data distribution by allocating data based on labels, ensuring that individual clients do not possess data with all labels. As shown in Figure~\ref{fig:model_accuracy}, our approach maintains accuracy compared to IFCA and performs better than FedAvg, demonstrating its ability to handle uneven data distributions and outperform traditional models.

Our method, which employs the weighted aggregation of cosine similarities derived from server updates and incorporates gradient scaling, has its model convergence performance intrinsically tied to the selection of the server update, denoted as $g_0$. In our approach, $g_0$ and subsequent model training should use the same model. If the model is trained from random initialization, the obtained gradients are likely to be classified as poisonous. This is because the initial model determines the starting position, and different starting positions result in different directions toward the approximate optimal solution.
Another factor is the weighting and scaling, as the learning rate significantly impacts convergence performance, and both weighting and scaling are affected by $\|g_0\|$. Our experiments showed that adjusting the training process of $g_0$ can improve overall performance. Additionally, multiplying by the $\|g_0\|$ factor may not always benefit the overall model. Essentially, the role of $\|g_0\|$ in Eq.\eqref{eq:rfca_aggregation} is to scale and counteract scaling attacks. However, in practice, $\|g_i\|$ already achieves this effect, so the choice of scaling $\|g_0\|$ is based on practical considerations or decreased to the learning rate.
Therefore, choosing the learning rate is crucial. In our experiments, to standardize the variable, we used dynamic learning rates to eliminate the influence of scaling, making it consistent with other approaches.

\begin{figure}[ht]  
  \centering  
  \subfloat[MNIST (n=10)]
    {\includegraphics[width=0.25\textwidth]{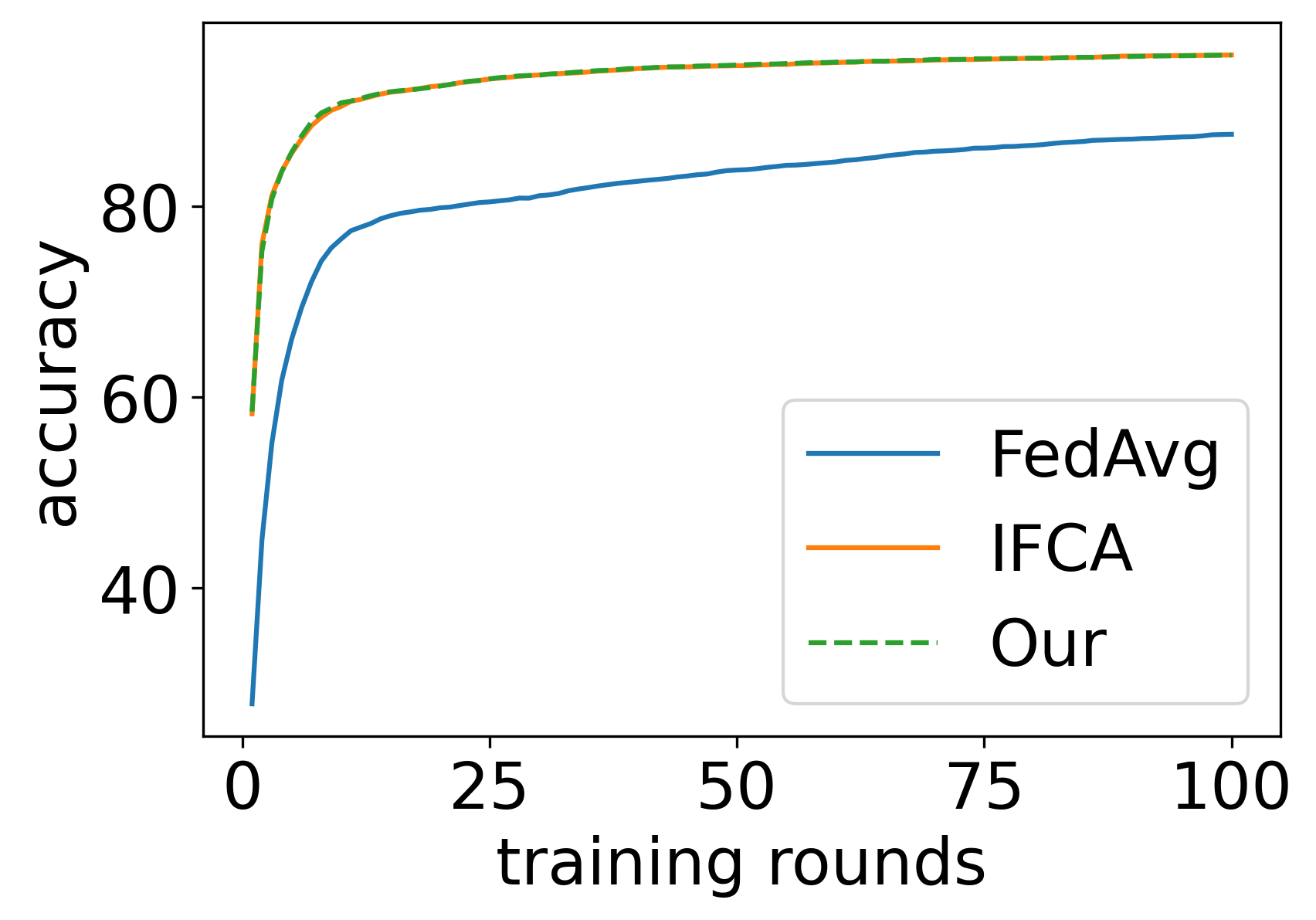}}
    \subfloat[CIFAR10 (n=10)]
    {\includegraphics[width=0.25\textwidth]{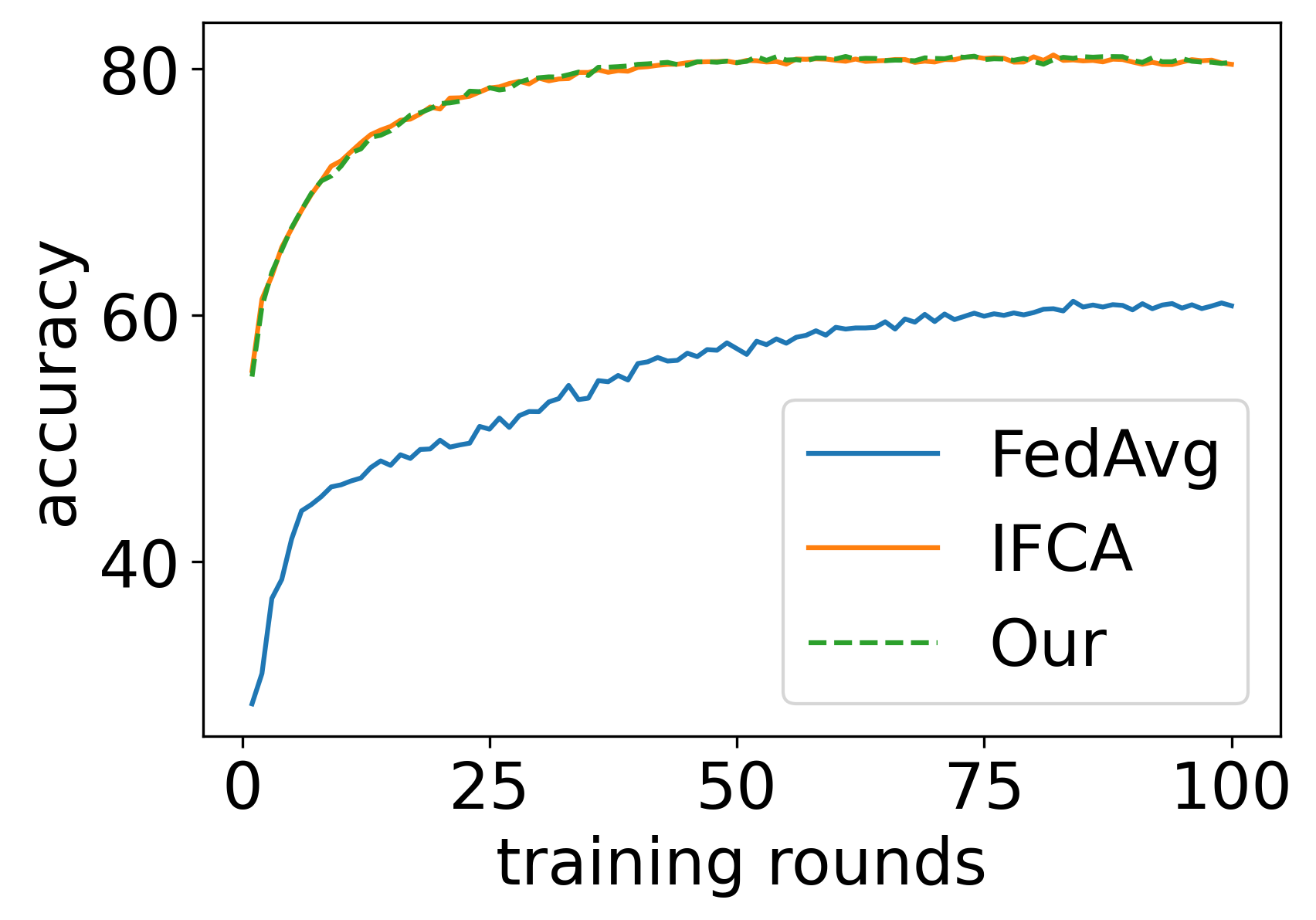}}
      
    \subfloat[MNIST (n=20)]
    {\includegraphics[width=0.25\textwidth]{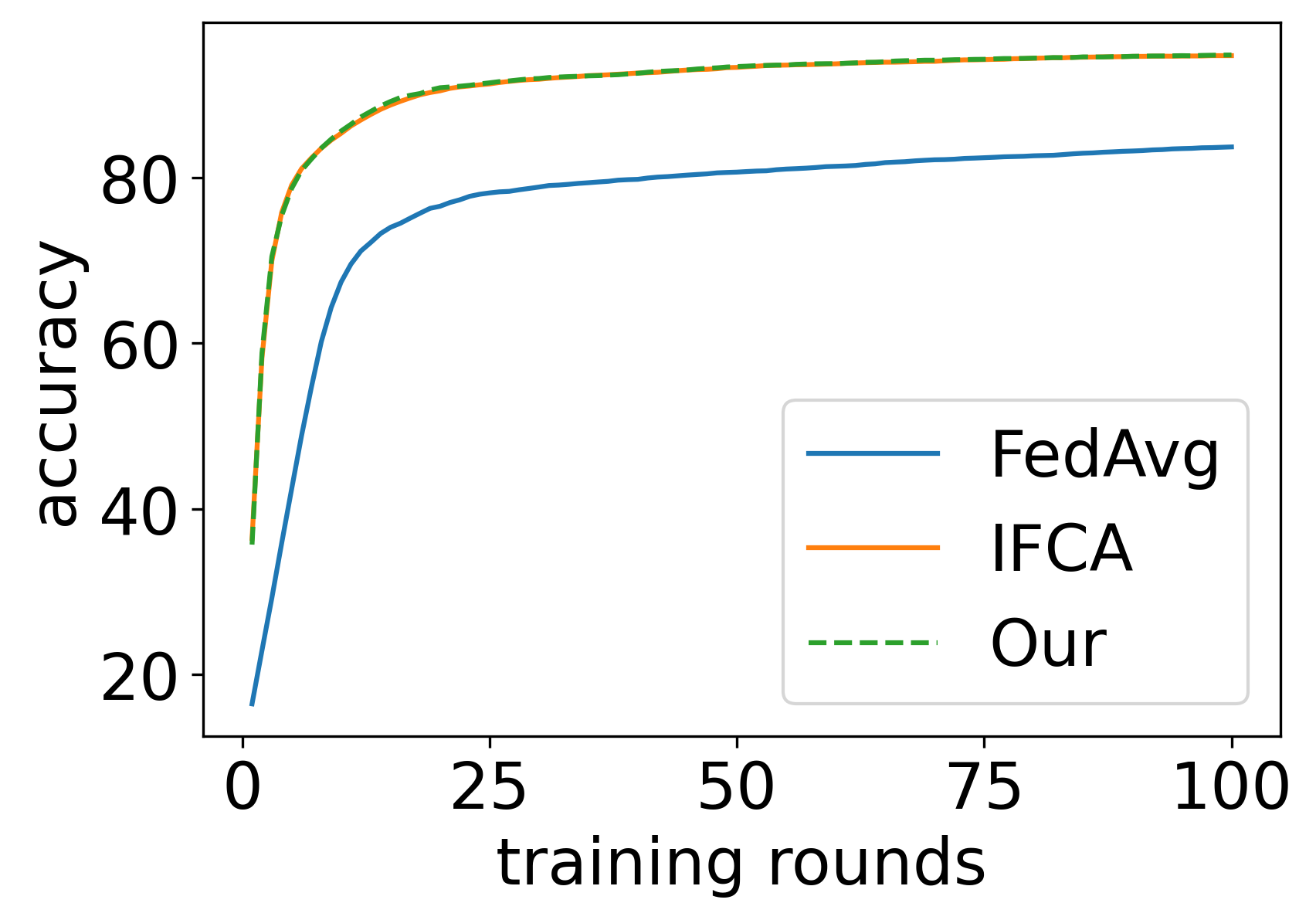}}
    \subfloat[CIFAR10 (n=20)]
    {\includegraphics[width=0.25\textwidth]{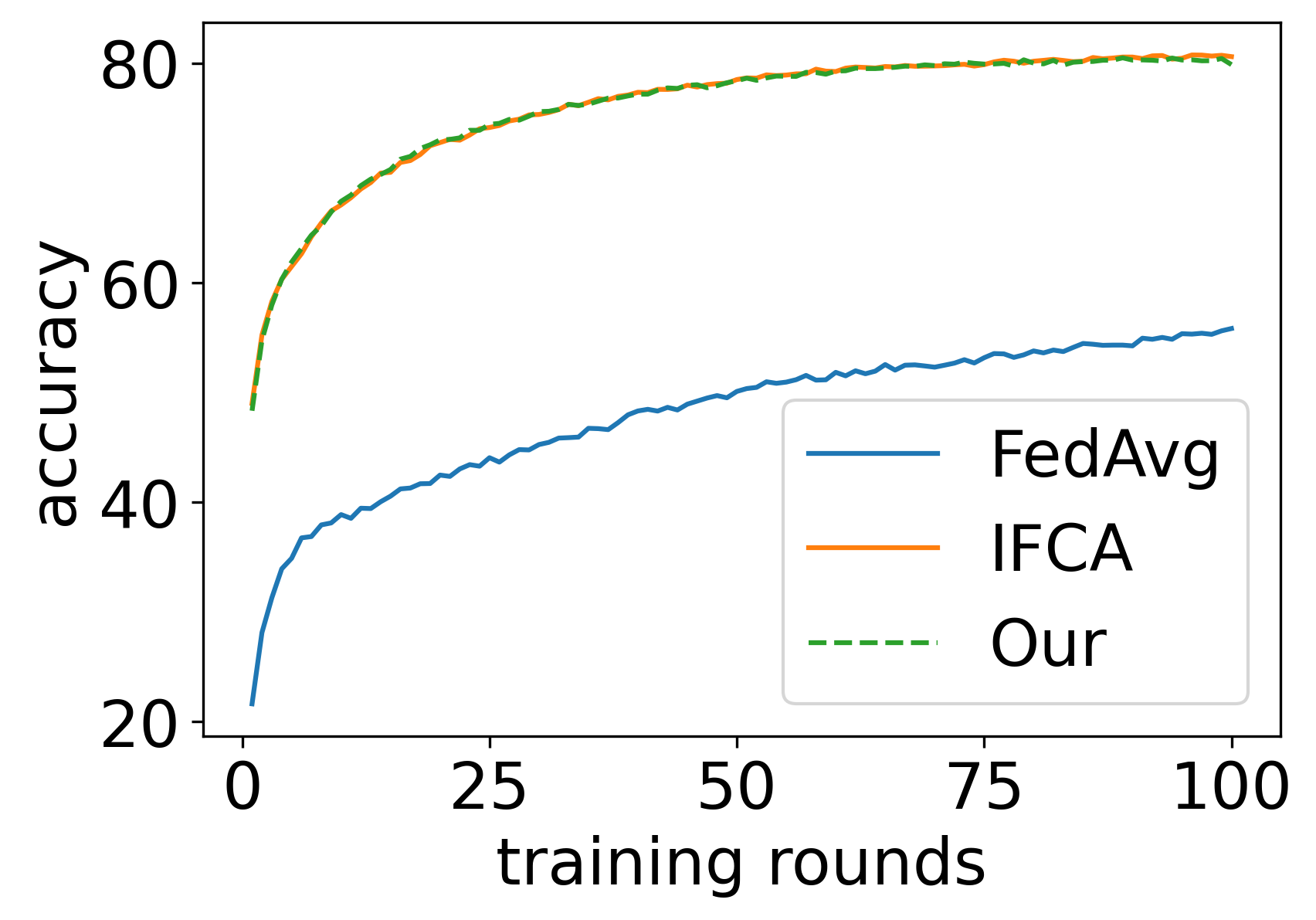}}
  \caption{The performance of the model on different datasets with varying numbers of clients.}
  \label{fig:model_accuracy}  
\end{figure} 
\subsection{Byzantine-robustness}
We selected representative and recent Byzantine attacks for evaluation, including Label-Flipping (LF) Attack, Krum Attack, Trim Attack, Scaling Attack, and Adaptive Attack.

Under the independent and identically distributed (i.i.d.) data setting, we compared our scheme with other representative schemes. The experimental results are shown in Table~\ref{tab:byzantine_attacks}. It is obvious that, our scheme maintains high accuracy under various attacks, showing a clear advantage over other schemes. In the experiments with the MNIST dataset, our results are similar to those of FLTrust. However, for the CIFAR10 dataset, our scheme achieves higher accuracy than FLTrust. This is because our scheme incorporates the features of IFCA, performing better than FedAvg under equivalent conditions. Since CIFAR10 is more complex than MNIST, this difference is more noticeable.

\begin{table*}[ht]
\small
\centering
    \begin{tabular}{ccccccccccc}
      \hline  
      ~ & \multicolumn{5}{c}{\textbf{MNIST}}  &  \multicolumn{5}{c}{\textbf{CIFAR10}}\\
      \textbf{Approach} & \textbf{LF} & \textbf{Krum} & \textbf{Trim} & \textbf{Scaling} & \textbf{Adaptive} & \textbf{LF} & \textbf{Krum} & \textbf{Trim}  &  \textbf{Scaling} & \textbf{Adaptive}\\  
      \hline  
      FedAvg & 92.21\% & 97.31\% & 37.54\% & 43.79\% & 96.84\% & 71.31\% & 70.93\% & 18.72\% & 9.86\% & 12.38\% \\  
      Median & 96.47\% & 97.25\% & 57.84\% & 97.06\% & 52.22\% & 68.54\% & 36.02\% & 24.67\% & 70.42\% & 17.33\% \\  
      Krum & 65.12\% & 5.32\% & 77.16\% & 67.55\% & 73.07\% & 46.63\% & 11.37\% & 40.68\% & 51.77\% & 40.92\% \\  
      FLTrust & 96.82\% & 98.31\% & 97.30\% & 96.52\% & 97.76\% & 64.29\% & 65.79\% & 66.17\% & 65.34\% & 66.23\% \\
      EBS-CFL & 97.82\% & 98.79\% & 97.21\% & 97.96\% & 98.43\% & 72.19\% & 71.76\% & \textbf{72.84\%} & 72.52\% & \textbf{71.46\%} \\  
      \hline  
    \end{tabular} 
\caption{Comparing the accuracy of different approaches under Byzantine attacks. LF attack, Krum attack, Trim attack, Scaling attack, and Adaptive attack follow the settings in research \cite{FLTrust}.}
\label{tab:byzantine_attacks}
\end{table*}

\begin{figure*}[ht] 
    \centering  
    \includegraphics[width=1\textwidth, clip]{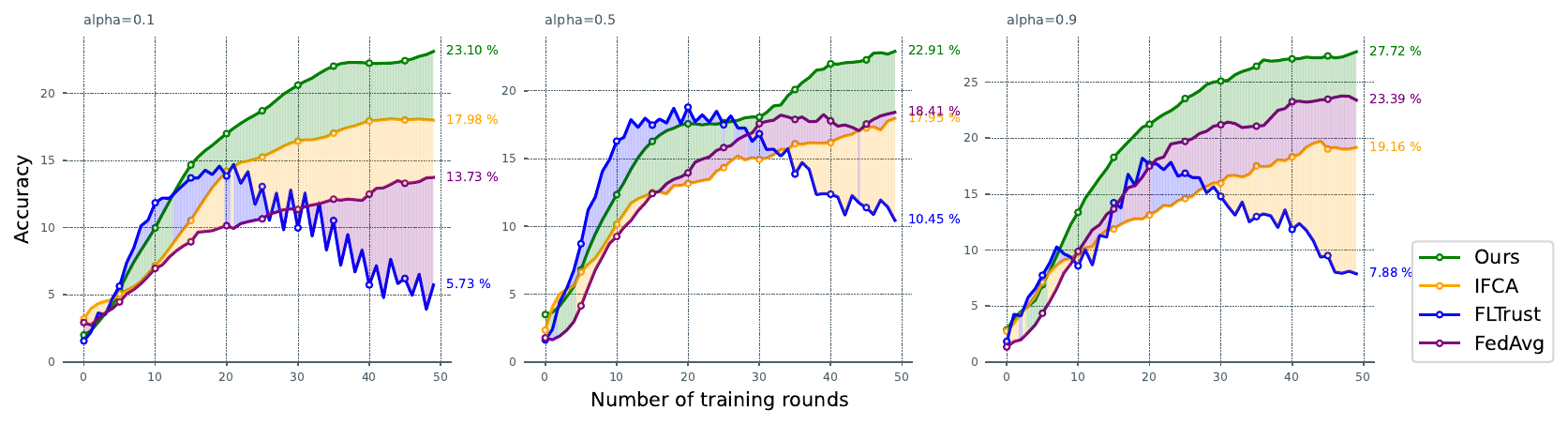}  
    \caption{Comparing the performance of different methods under LF attacks in a non-i.i.d. setting on the CIFAR100 dataset.}  
    \label{fig:CIFAR100_LF}
\end{figure*}

Under the non-i.i.d. data setting, we conducted additional experiments on the CIFAR100 dataset, keeping all experimental settings consistent.

As shown in Figure~\ref{fig:CIFAR100_LF}, our method consistently outperforms the compared methods across different $\alpha$ settings. Additionally, it can be observed that FLTrust exhibits an initial increase followed by a decrease in performance in all experiments. This occurs because, in non-i.i.d and more complex datasets, it becomes increasingly difficult to distinguish between the adversarial poisoned gradients and the benign gradients submitted by clients. Meanwhile, the adversary's influence on the model gradually increases. FLTrust updates the server with the newly updated global model each round and uses this model as a basis to evaluate the gradients. This means that in cases where the global model is biased, the poisoned gradients may receive a higher aggregation weight, ultimately reducing model accuracy. Therefore, FLTrust may have an advantage early when the global model is less affected, but as the global model shifts, the algorithm inadvertently amplifies the LF attack. In contrast, our proposed EBS-CFL generates server updates during the initialization phase, making its gradient evaluation independent of the current global model, and thus, its accuracy curve does not exhibit the rise-and-fall phenomenon.

Furthermore, it is obvious that IFCA outperforms FedAvg when $\alpha=0.1$. However, when $\alpha \le 0.5$, FedAvg's accuracy exceeds IFCA's. This pattern emerges due to the fact that at lower $\alpha$ levels, the non-i.i.d effects significantly impede FedAvg. As $\alpha$ rises, the impact of non-i.i.d diminishes. Meanwhile, the LF attack exerts a greater influence on IFCA than FedAvg, resulting in its inferior performance compared to FedAvg. Our proposed EBS-CFL combines the strengths of both methods and includes Byzantine-robustness, resulting in superior performance in all experiments.

\section{Efficiency Analysis}

We analyze the communication and computation efficiency of the compression scheme for each round, excluding the overhead of initialization, which is performed only once at the beginning.
Analysis results are provided in Table~\ref{tab:communication_computation}.

\subsection{Communication Overhead}
We analyze the communication overhead from the perspectives of both the client and the server.

\textbf{Client:} The communication complexity for the client involves two components, i.e., the key sent by the Trusted Authority (KDC) and the encoded gradients sent to the server. For the $i$-th client, the key $ek_i = \{ek_i^0, ek_i^1\}$, where $ek_i^0 \in \mathbb{R}^{tm \times 3tm\xi_1}$ and $ek_i^1 \in \mathbb{R}^{s \times 3tm\xi_1}$. The encoded gradient is denoted as $\delta_i \in \mathbb{R}^{s \times 3tm\xi_1}$. Here, $\xi_1$ and $t$ are constants, and $s$ can be expressed as $\frac{l}{t}$, resulting in a communication complexity for the client is $O(m^2 + ml)$.

\textbf{Server:} The server's communication complexity includes the verification parameters sent by the KDC, $\{\mathcal{V}_i \cdot \sum_{j=1}^m {R''_{ij}}^T\}_{i=1}^n$, and the encoded gradients $\delta_i$ sent by the clients. For $1 \le i \le n$, $\mathcal{V}_i \cdot \sum_{j=1}^m {R''_{ij}}^T \in \mathbb{R}$. With $n$ clients, the server's communication complexity is $O(nml + n)$.

\subsection{Computation Overhead}
We analyze the computation overhead from both perspectives of client and server .

\textbf{Client:} The client's computation involves encoding gradients. For the $i$-th client, $1 \le j \le s$, the gradient fragment $g_{i,j} \in \mathbb{R}^{1 \times t}$ is encoded as $E(g_{i,j}) = g_{i,j}ek_i^0 + ek_{i,j}^1$, with a computational complexity of $O(m^2l)$.

\textbf{Server:} The server's computation includes gradient verification, SRFC execution, and hierarchical aggregation.

{\it Gradient Verification:} This involves Eq.\eqref{eq:verify} and 
$\forall i \in [n], \; ||\delta_i||^2 = 3$, both with a complexity of $O(nl)$.

{\it SRFC Computation:} Before running SRFC, for $1 \le i \le n$, $1 \le j \le 2tm\xi_1$, and $1 \le k \le s$, we first compute $\alpha_i$, where $\alpha'_{j,k} \in \mathbb{R}^{3tm\xi_1 \times 2tm\xi_1}$. The combined computational complexity of calculating $\alpha_i$ and the inner product between clients' gradients and server updates is $O(nm^3l)$. SRFC is computed based on Eq.\eqref{eq:ReLU_encode} and Eq.\eqref{eq:ReLU_decode}. Here, $\alpha_i \in \mathbb{R}^{2tm\xi_1 \times 2tm\xi_1}$, $\beta \in \mathbb{R}^{tm \times 2tm\xi_1}$, and $\tau_1, \tau_2, \tau_3 \in \mathbb{R}^{2tm\xi_1 \times 2tm\xi_1}$. The computational complexity of calculating $\beta \cdot (\sum_{i=1}^n \alpha_i)$ is $O(nm^2)$. The computational complexity of calculating     
$\sum_{i=1}^n \mathcal{F}^{-1}(\beta \cdot \alpha_i^2 \cdot \tau_2 + \mathcal{F}^{-1}(\beta \cdot \alpha_i^2 \cdot \tau_3))$ is $O(nm^4)$.

{\it Layered Aggregation:} For hierarchical aggregation and transformation of $\delta$, we set $\xi_1 = 1$ and for $2 \le i \le x$, $\xi_2 = 2$. The transformation key is $tk_{i-1} \in \mathbb{R}^{3tm\xi_{i-1} \times 3tm\xi_i}$, with a computational complexity of $O(m^2nl \log{n})$. For the transformation of $\alpha$, we set $x = 2$, $\xi_1 = 1$, $\xi_2 = n$. The transformation key is $tk_i \in \mathbb{R}^{3tm\xi_1 \times 3tm\xi_2}$, with a computational complexity of $O(m^2n^2)$.

In summary, the server's computational complexity is $O(nm^3l + nm^4 + m^2n^2 + m^2nl\log{n})$.

\section{Pseudo Code of RFCA}
This section provides a detailed description of the pseudo code for RFCA as Algorithm~\ref{alg:rfca}. The RFCA is designed to enhance the Byzantine-robustness of federated learning in non-i.i.d. environments. It begins by initializing cluster-specific models on the server using a root dataset, which is partitioned to reflect real-world data distributions. During each global iteration, the server broadcasts the global models to all clients. Clients then perform local updates based on their data, determining the most suitable cluster for their model. They send their updated gradients and cluster information back to the server. The server aggregates these gradients, adjusting the cluster models using a weighted scheme based on the alignment of client gradients with the initial server updates, reinforcing the Byzantine-robustness against adversarial attacks.

\begin{algorithm}[tb]
\caption{Robust Federated Clustering Learning Algorithm (RFCA)}  
\label{alg:rfca}
\small
\begin{algorithmic}[1]  
\Statex \textbf{Input:} number of clusters $m$; learning rate $\eta$; learning rate of initialization $\eta_I$; initialization $\{\theta^0_i\}_{j=1}^m$; $n$ clients with local training datasets $D_i$, for $1 \le i \le n$; a server with root dataset $D_0$; number of global iterations $T_g$; number of server iterations $T_s$; number of local iterations $T_l$; batch size $b$.
    \Statex \underline{Server initialization: }  
    \State Split $D_0$ into $\{D_0^i\}_{i=1}^n$ according to real-world data distribution.
    \For{$j=1$ to $m$}
        \State $g_0^j \gets 0$
    \EndFor
    \For{$i=1$ to $m$}
        \State $j, \hat{g} \gets ClusteredModelUpdate(\{\theta^0_i\}_i, D_0, b, \eta_I, T_s)$
        \State $g_0^j \gets g_0^j + \frac{m}{n}\hat{g}$
    \EndFor
    
    \For{$t=1$ to $T_g$}
        \Statex \underline{Server: } 
        
        The server broadcast $\{\theta^0_i\}_{j=1}^m$ to clients
        \Statex \underline{Client: } 
        \For{$i=1$ to $n$}
            \State $\hat{j}, g_i \gets ClusteredModelUpdate(\{\theta^0_i\}_{j=1}^m, D^i, b, \eta, T_l)$
            \State define one-hot encoding vector $s_i \gets \{s_{i,j}\}_{j=1}^m$ ($s_{i,j} = 1$ where $j = \hat{j}$)
            \State send $\{s_{i,j}\}_{j=1}^m, g_i$ to the server.
        \EndFor
        \Statex \underline{Server: }  
        \For{$i=1$ to $n$}
            \For{$j=1$ to $m$}
                \State $c^j_i \gets s_{i,j} ReLU(\frac{\langle g_i, g_0^j \rangle}{\lVert g_i \rVert \cdot \lVert g_0^j \rVert})$
            \EndFor
        \EndFor
        \State $c_{agg} \gets \sum_{i=1}^{n}\sum_{j=1}^{m} \text{ReLU}(c^j_i)$
        \For{$j=1$ to $m$}
            \State $\theta^{t}_j \gets \theta^{t-1}_j + \eta \frac{1}{c_{agg}} \sum_{i=1}^n \text{ReLU}(c^j_i) \frac{\lVert g_0^j \rVert}{\lVert g_i \rVert} g_i$
        \EndFor
    \EndFor  
    \State \Return $\{\theta^T_i\}_{i=1}^m$
\end{algorithmic}  
\end{algorithm}

\section{Proof}
\subsection{Proof of Theorem \ref{th:verify}}
\label{pf:verify}
We define $\{x_i\}_{i=1}^n$ which satisfy that $\sum_{i \in n}x_i = \sum_{i \in n}\chi_i \sum_{i \in n}M_i^T$, so we have
\begin{equation*}
    \sum_{i=1}^n a_i \cdot dk^T \ne \sum_{i=1}^n x_i \iff \sum_{i \in n} R'_i \ne \sum_{i \in n} R_i.
\end{equation*}
And we have
\begin{equation*}
    \begin{aligned}
        &\sum_{i=1}^n a_i \cdot dk^T \ne \sum_{i=1}^n x_i \wedge \forall i \in [n], \mathcal{V}(\chi_i, vk)=1 \iff \\ 
        &\sum_{i \in n} R'_i \ne \sum_{i \in n}R_i \wedge \forall i \in [n], \mathcal{V}(\chi_i, vk)=1 \Rightarrow \\ 
        & (\exists i \in [n])(R'_i \ne R_i \wedge \mathcal{V}(\chi_i, vk)=1).
    \end{aligned}
\end{equation*}
When $a_i = \chi_i'$, let $a_i = \sum_{j \in S} x_j'M_j + \sum_{j \in S} R'_jM_{j+k}$. 

So the left term of Eq.~\ref{eq:p_verify} can derive that
\begin{equation*}
    \begin{aligned}
        & Pr[\sum_{i \in n} R'_i \ne \sum_{i \in n}R_i \wedge \forall i \in [n], \mathcal{V}(\chi_i, vk)=1] \\
        \le & Pr[(\exists i \in [n])(R'_i \ne R_i \wedge \mathcal{V}(\chi_i, vk)=1)] \\
        =&Pr[\mathcal{V}_iR^T=0] 
        = \frac{1}{\lambda^l} = \epsilon,
    \end{aligned}
\end{equation*}

where $R$ is the attacker's random guessing matrix.
\subsection{Proof of Theorem \ref{th:srfc}}
\label{pf:srfc}
The terms of Eq.\eqref{eq:ReLU_encode} can be expressed as
\begin{equation}
\label{eq:t_1}
    \begin{aligned}
        &\beta(\sum_{i=1}^n \alpha_i)\tau_1
        =\beta \sum_{i=1}^n x_iM_i^TM_i + M_{i+n}^T(R_i^{(0)} + \zeta_i \cdot M_{i+n})\\
        &= \sum_{i=1}^n x_iM_i + R_i^{(0)} + \zeta_i \cdot M_{i+n},
    \end{aligned}
\end{equation}
and
\begin{equation}
\label{eq:t_2}
    \begin{aligned}
        &\mathcal{F}^{-1}(\beta\alpha_i^2\tau_3)
        = \mathcal{F}^{-1}(x_i^2M_i^T\mathcal{F}(2M_i \circ R_i^{(1)})
        - M_{i+n}^T\mathcal{F}(R_i^{(2)})) \\
        & = 2x_iM_i \circ R_i^{(1)} - R_i^{(2)}.
    \end{aligned}
\end{equation}
According to \eqref{eq:t_2}, we have
\begin{equation}
\label{eq:t_3}
    \begin{aligned}
        & \mathcal{F}^{-1}(\beta \alpha_i^2 \tau_2 + \mathcal{F}^{-1}(\beta \alpha_i^2 \tau_3))
        = \mathcal{F}^{-1}(x_i^2 (M_i \circ M_i) \\
        &- R_i^{(1)} \cdot R_i^{(1)} + R_i^{(2)} + \mathcal{F}^{-1}(\beta\alpha_i^2\tau_3))\\
        &= \mathcal{F}^{-1}(x_i^2 (M_i \circ M_i) - R_i^{(1)} \cdot R_i^{(1)} + 2x_iM_i \circ R_i^{(1)})\\
        &= |x_i|M_i + R_i^{(1)}.
    \end{aligned}
\end{equation}
According to \eqref{eq:t_1} and \eqref{eq:t_3}, the right term of Eq.\eqref{eq:ReLU_encode} can be expressed as
\begin{equation}
    \begin{aligned}
        &\frac{1}{2}(\sum_{i=1}^n x_iM_i + R_i^{(0)} + \zeta_i \cdot M_{i+n} + |x_i|M_i + R_i^{(1)})\\
        &= \sum_{i=1}^n ReLU(x_i)M_i + \frac{1}{2} \zeta_i \cdot M_{i+n} = \sum_{i=1}^n E(ReLU(x_i)).
    \end{aligned}
\end{equation}
So the Eq.\eqref{eq:ReLU_encode} is true. And for Eq.\eqref{eq:ReLU_decode}, we have
\begin{equation}
    \begin{aligned}
        & \frac{1}{l}tr(\sum_{i=1}^n E(ReLU(x_i) \cdot \beta^T) = \frac{1}{l}tr(\sum_{i=1}^n ReLU(x_i)I + \frac{1}{2} \zeta_i)\\
        &= \frac{1}{l}\cdot l \cdot \sum_{i=1}^n ReLU(x_i) = \sum_{i=1}^n ReLU(x_i).
    \end{aligned}
\end{equation}
\section{Security Analysis}
\begin{definition}
    The perturbation function $\phi$ is defined as follows.
    \begin{equation}
        \phi(x, M) = x \cdot M,
    \end{equation}
   where $M \cdot M^T = I$, and it satisfies
   \begin{equation*}
        Pr(x'=x|\phi(x, M)) = Pr(M'|M' \cdot M^T = I),
    \end{equation*}
    where $M'$ presents random guess Matrix.
\end{definition}

\begin{definition}
    The perturbation function $\psi$ is defined as follows.
    \begin{equation}
        \psi(x, R) = x + R,
    \end{equation}
    where $x \cdot R^T \ne O$ and $R^T \cdot x \ne O$, it satisfies
    \begin{equation*}
        Pr(x'=x|\psi(x, R)) = Pr(R' = R).
    \end{equation*}
    where $R'$ presents random guess Matrix.
\end{definition}

\begin{definition}
    The function $\varphi$ is defined as follows.
    \begin{equation}
        \varphi(\mathcal{S}) = \sum_{M \in \mathcal{S}} M,
    \end{equation}
   where $\mathcal{S} \subset \{M_i\}_{i=1}^n$, and $\{M_i\}_{i=1}^n$ are Mutually Orthogonal Matrices.
\end{definition}

\begin{theorem}
\label{th:sum_mo}
    The function $\varphi$ satisfies that
    \begin{equation}
    \label{eq:sum_mo_ineq}
        \text{Pr}[M' \in S|\varphi(\mathcal{S})] \le \text{negl}(n).
    \end{equation}
\end{theorem}  

\begin{proof} 
    Due to the summation of all elements in $S$ by $\varphi$, it is impossible to obtain individual elements in $S$.

    Considering whether there exists $S'$, different from $S$, such that the elements in $S'$ form pairwise orthogonal matrices, and $|S| = |S'|$, we have

    \begin{equation}
    \label{eq:sum_mo}
        \varphi(\mathcal{S}) = \varphi(\mathcal{S'}).
    \end{equation}
    Taking a three-dimensional orthogonal matrix as an example, a $3$-dimensional orthogonal matrix can be divided into a set of pairwise orthogonal matrices $\{M_i\}_{i=1}^3$. For $1 \le i \le 3$, $M_i \in \mathbb{R}^{1 \times 3}$, so $M_i$ can be viewed as a vector.
 
    \begin{figure}[ht]  
    \centering  
    \includegraphics[width=0.5\textwidth]{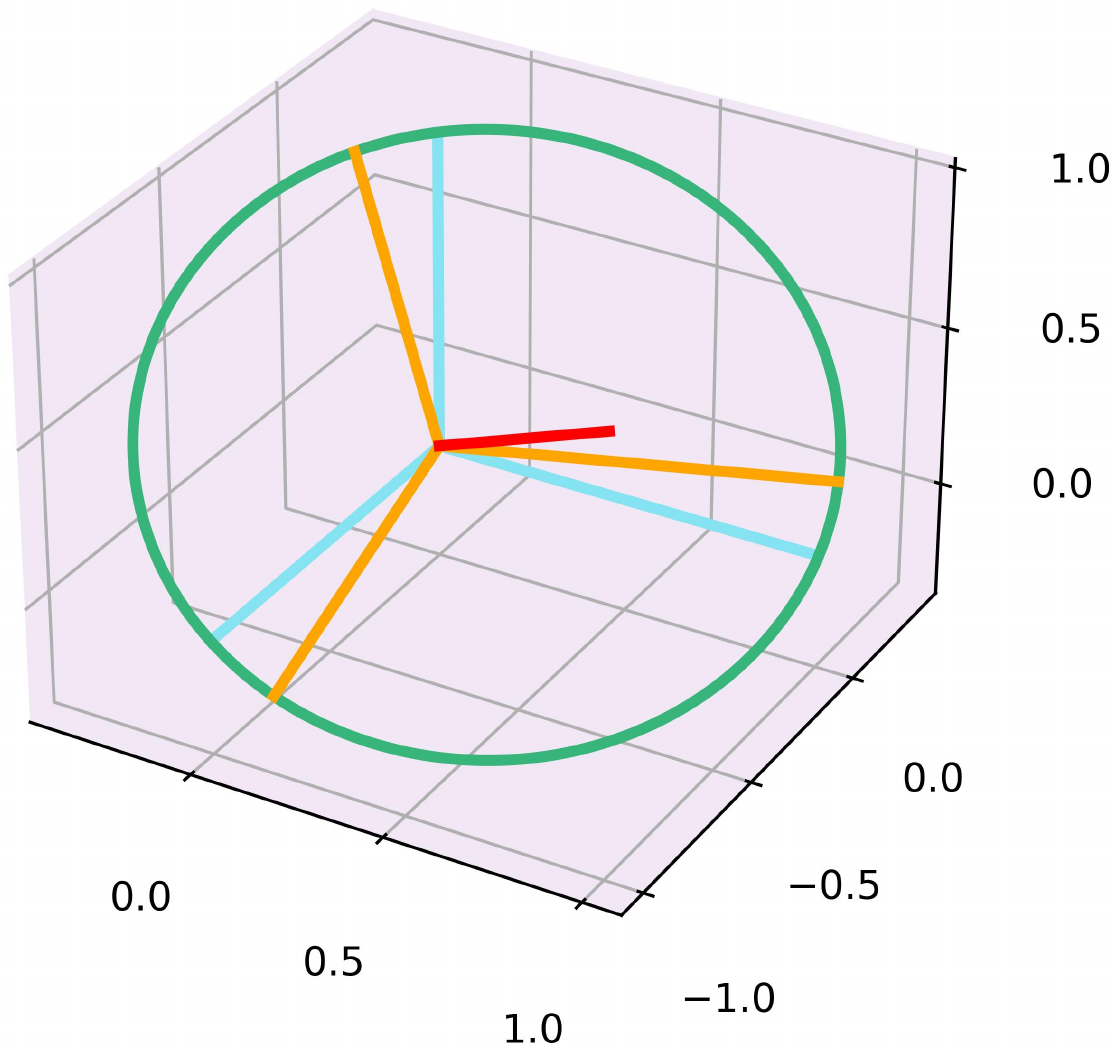}  
    \caption{Example of mutually orthogonal matrices rotation.}  
    \label{fig:sum_mo}
    \end{figure}
    
    As shown in the Figure~\ref{fig:sum_mo}, the blue vectors represent elements in $S$, while the orange vectors represent elements in $S'$. The red vector is the result of summing the blue vectors, and the green trajectory is formed by rotating the blue vectors around the red vector as an axis. The orange vectors are obtained by rotating the blue vectors, and their endpoints lie on the green trajectory. The orange vectors can still form pairwise orthogonal matrices, satisfying Eq.\eqref{eq:sum_mo}. Obviously, there exist countless rotations in the real number field that result in sets $S'$ similar to the orange vectors.

    For any set $S$, assuming $\forall M \in S, M \in \mathbb{R}^{m \times n}$, for $1 \le i \le m$, $M \in S$,
    \begin{equation}
    \label{eq:sum_row}
        \vec{a_i} = \sum_{M \in S} M_i,
    \end{equation}
    where $M_i$ represents a vector composed of the elements in the $i$-th row of $M$. 
    Obtain the rotated vector $M'_i$ by rotating the corresponding row vector around $\vec{a_i}$, and correspondingly form $M'$, finally obtaining $S'$. 
    Since $M'_i$ can be expressed by $M_i$, and for $1 \le i,j \le m, \; i \ne j, \; M_i \cdot M_j = O$, therefore $M'_j \cdot M'_j = O$ holds.

    In this constructed set $S'$, the matrices still form pairwise orthogonal matrices, and $|S| = |S'|$, while satisfying Eq.\eqref{eq:sum_mo}.

    Therefore, there exist sets $S'$ different from $S$, and in the real number field, there exist countless such sets $S'$ satisfy these conditions. This implies that the elements composing $\varphi(\mathcal{S})$ are not unique, and there are countless solutions in the real number field, thus Eq.\eqref{eq:sum_mo_ineq} holds.
\end{proof}

\begin{theorem}
    In the VOMCA, individual variable is confidential.
\end{theorem}

\begin{proof}
    For $1 \le i \le n$, $\chi_i$ involves the random number $R_i$, therefore, to ascertain $x_i$, it is necessary to eliminate $R_i$. However, $dk$ contains $M_{i+n}$, making it impossible to eliminate $R_i$.
    
    On the other hand, it is also unfeasible to achieve this through $vk$. Because $\mathcal{V}_i$ is unknown and irreversible, it is impossible to eliminate $\mathcal{V}_i$ from $vk$ to obtain $M_{i+n}$.

    According to theorem~\ref{th:sum_mo}, it's hard to obtain $M_i$ from $dk$ and $vk$.

    Therefore, in the VOMCA, individual variable is confidential.
\end{proof}

\begin{theorem}
\label{th:sec_srfc}
    In the SRFC Mechanism, individual variable is both confidential and immutable.
\end{theorem}

\begin{proof}
    According to theorem~\ref{th:sum_mo}, for $1 \le i \le n$, it's hard to obtain $M_i$ from $\beta$, so it's hard to deduce $x_i$ from $\alpha_i$ because $M_i$ is unknown, thus the interference of $R'_i$ cannot be eliminated from the overall context. 
    
    During the computation process in Eq.\eqref{eq:ReLU_encode}, we can refer to theorem \ref{th:srfc} to examine the intermediate variables. 
    
    The random matrix $R_i^{(2)}$ is a sparse matrix, where some of its elements are zero when the corresponding positions in $R_i^{(1)}$ are non-zero. Evidently, it's numerically challenging to discern whether the results or intermediate outcomes in equations \eqref{eq:t_2} and \eqref{eq:t_3} are derived from $R_i^{(1)}$ or $R_i^{(2)}$.
    
    It can be observed that Eq.\eqref{eq:t_1}, Eq.\eqref{eq:t_2} and Eq.\eqref{eq:t_3} involve random numbers. Given these random numbers cannot be eliminated before the summation process, it precludes the possibility of solving for individual variables.

    According to the results of Eq.\eqref{eq:t_2} and Eq.\eqref{eq:t_3}, for $1 \le j \le lm$, $1 \le k \le 2lmn$, if $(R_i^{(1)})_{jk} \ne 0$,

    \begin{equation*}
        2x_i(M_i)_{ij} (R_i^{(1)})_{ij} - (R_i^{(2)})_{ij} \in (-2|x_i|^2, 0)
    \end{equation*}
    \begin{equation*}
        |x_i| \cdot (Mi)_{jk} + (R_i^{(1)})_{jk} \in (-|x_i|, |x_i|),
    \end{equation*}
    if $(R_i^{(1)})_{jk} = 0$,
    % 如果 $(R_i^{(1)})_{jk} = 0$
    \begin{equation*}
        2x_i \cdot (M_i)_{ij} (R_i^{(1)})_{ij} - (R_i^{(2)})_{ij} \in (-2|x_i|^2, 0),
    \end{equation*}
    \begin{equation*}
        |x_i| \cdot (Mi)_{jk} + (R_i^{(1)})_{jk} \in (0, |x_i|).
    \end{equation*}
    Because the distribution of the results from Eq.\eqref{eq:t_2} is the same, it's only need to discuss the results of Eq.\eqref{eq:t_3}.
    According to Eq.\eqref{eq:random_sample_matrix1} and Eq.\eqref{eq:random_sample_matrix2}, approximately $l^2m^2n$ elements in the matrix have $(R_i^{(1)})_{jk}=0$. To isolate $R_i^{(2)}$, it is necessary to correctly guess all the elements where $(R_i^{(1)})_{jk}=0$. According to probability, approximately there are
    \begin{equation*}
        (\frac{1}{2} \cdot \frac{1}{2} + \frac{1}{2}) \cdot 2l^2m^2n = \frac{3}{2}l^2m^2n
    \end{equation*}
    elements in the matrix in $(0, |x_i|)$.
    Therefore, the probability of isolating $R_i^{(2)}$ based on blind guessing is $\frac{1}{2}^{\frac{3}{2}l^2m^2n}$.Therefore, eliminating $R_i^{(2)}$ proves to be challenging, implying the need to utilize Eq.\eqref{eq:t_3} for its removal. Since Eq.\eqref{eq:t_1} includes $R_i^{(0)}$ and Eq.\eqref{eq:t_3} includes $R_i^{(1)}$, it is clear that only by combining the results of Eq.\eqref{eq:t_1} and Eq.\eqref{eq:t_3} can we eliminate $R_i^{(0)}$ and $R_i^{(1)}$.
    
    On the other hand, since all equations involve random numbers, this means that multiplication also affects random numbers. If operations are performed only on certain variables, then it will result in the sum of random numbers not being zero, ultimately making decryption impossible. 
    
    Therefore, in the SRFC scheme, individual variables are both confidential and immutable.
\end{proof}

\begin{theorem}
    In the secure aggregation scheme, the gradients and the cluster identity of each client are confidential.
\end{theorem}

\begin{proof}
    For all $i$ satisfying $1 \le i \le n$, we cannot directly extract $g_i$ and its cluster identity from $\delta_i$. If we know these two pieces of information, it is equivalent to knowing $M'_i$. According to $dk$, knowing the information of $M'_i$ is equivalent to knowing $M_i$. Obviously, we cannot obtain this information through $\alpha'_i$, nor can we obtain it through the parameters generated in SRFC. 
    
    On the other hand, the server cannot obtain $g_i$ by altering the aggregation weights of a client and obtaining the differential aggregation multiple times. This can be explained by Theorem \ref{th:sec_srfc}.

    Therefore, in the secure aggregation scheme, the gradients and the cluster identity of each client are confidential.
\end{proof}

\begin{theorem}
\label{th:skc}
    In the SKT Mechanism, for $1 \le i \le n$, the $x_i$, the keys $M_i$ and $M_{i+n}$, and the conversed keys $M'_{\mathcal{C}(i)}$ and $M'_{\mathcal{C}(i+n)}$ are confidential.
\end{theorem}

\begin{proof}
    For $1 \le i \le n$, $1 \le j \le m$, given $tk_i$ and $\chi_i$, obtaining $x_i$ is equivalent to knowing $M_i$. Grasping $M_i$ is equivalent to knowing $M'_{\mathcal{C}(i)}$, which, in turn, can be inferred from $M_{i+n}$. Due to the random number $R'_i$, it is impossible to infer $M_{i+n}$ from $M'_{\mathcal{C}(i)}$. Evidently, the information cannot be retrieved via $dk$ as it is structured as a summation, thereby preventing its separation.

    Further, consider a specific case: given the knowledge of $\mathcal{C}$, let's assume $a \neq b$ and $\mathcal{C}(a) = \mathcal{C}(b)$. In this scenario, an adversary can deduce $M'_{\mathcal{C}(a)}$ from $M_a$, but is unable to derive $x_b$ from $\mathcal{T}(\chi_i, tk_i)$ and $M'_{\mathcal{C}(a)}$, as $\mathcal{T}(\chi_i, tk_i)$ includes the random number $R'_b$ that can only be eliminated through summation. Once summed, individual data becomes obscured.

    Hence, within the SKT Mechanism, for every $1 \le i \le n$, the elements $x_i$, the keys $M_i$ and $M_{i+n}$, as well as the corresponding keys $M'_{\mathcal{C}(i)}$ and $M'_{\mathcal{C}(i+n)}$, are all maintained as confidential.
\end{proof}

\begin{theorem}
    In the compression scheme, the gradients and the cluster identity of each client are confidential.
\end{theorem}

\begin{proof}
    In Gradient Segmentation, it is impossible to obtain information about the keys across different segments of the same client because the random numbers added to each segment are different. Furthermore, the server cannot calculate the gradient changes of individual clients by subtracting gradients from different rounds, as $R''$ is updated before each round. 

    Regarding Layered Aggregation, for simplicity, we will not elaborate on the parts similar to the aggregation scheme. The security of the key transformation phase can be explained through the security of SKT as Theorem \ref{th:skc}, as this process does not expose information about the gradients or keys.

    Therefore, in the compression scheme, the gradients and the cluster identity of each client are confidential.
\end{proof}

\section{Convergence Guarantees}
In this section, we will provide a detailed convergence analysis for RFCA. Similar to the concepts in IFCA \cite{IFCA}, we define
\[
    F^j(\theta) := \mathbb{E}_{z \sim D_j}[f(\theta; z)],
\]

\[
    F_i(\theta; Z) = \frac{1}{|Z|} \sum_{z \in Z} f(\theta; z),
\]

where \( z \) denotes a data point, and \( Z \) represents the set of all such points.

\begin{definition}
If for any \( \theta, \theta' \), the function \( F \) satisfies the following condition.  
\[
F(\theta') \geq F(\theta) + \nabla F(\theta)^\top (\theta' - \theta) + \frac{\lambda}{2} \|\theta' - \theta\|^2,
\]  
\( F \) is said to be \( \lambda \)-strongly convex.
\end{definition}
\begin{definition}
If for any \( \theta, \theta' \), the function \( F \) satisfies the following condition.
\[
\|\nabla F(\theta') - \nabla F(\theta)\| \leq L \|\theta' - \theta\|,
\]  
\( F \) is said to be \( L \)-smooth.
\end{definition}

\begin{assumption}
For $j \in [k]$, $F^j(\theta')$ is $\lambda$-strongly convex and $L$-smooth.
\end{assumption}

\begin{assumption}
For every \( \theta \) and every \( j \in [k] \), the variance of \( f(\theta; z) \) is upper bounded by \( \eta^2 \) when \( z \) is sampled according to \( D_j \). Formally,  
\[
\mathbb{E}_{z \sim D_j} \left[ \left( f(\theta; z) - F^j(\theta) \right)^2 \right] \leq \eta^2.
\]
\end{assumption}

\begin{assumption}
For every \( \theta \) and every \( j \in [k] \), the variance of \( \nabla f(\theta; z) \) is upper bounded by \( v^2 \) when \( z \) is sampled according to \( D_j \). Formally,  
\[
\mathbb{E}_{z \sim D_j} \left[ \| \nabla f(\theta; z) - \nabla F^j(\theta) \|_2^2 \right] \leq v^2.
\]
\end{assumption}

\begin{assumption}
Assume that $\max_{j \in [k]} \|\theta_j^*\| \lesssim 1$.
We also assume that \( \|\theta_j^{(0)} - \theta_j^*\| \leq \left( \frac{1}{2} - \alpha_0 \right) \sqrt{\frac{\lambda}{L}} \Delta, \, \forall j \in [k]\)  
, \( n \gtrsim \frac{k \eta^2}{\alpha_0^2 \lambda^2 \Delta^4}\)  
, \( p \gtrsim \frac{\log(mn')}{m}\)  
, and that \( \Delta \geq O\left( \max \left\{ \alpha_0^{\frac{-2}{5}} (n')^{\frac{-1}{5}}, \alpha_0^{\frac{-1}{3}} m^{\frac{-1}{6}} (n')^{\frac{-1}{3}} \right\} \right)\).
\end{assumption}

\begin{assumption}
% 假设对于任意 $i \in [m]$ 余弦相似度计算得到的权重 $w_i$ 服从分布 $\mathcal{N}(\mu, \sigma^2)$， 其中 $0 < \mu < 1$ 且 $\sigma < 1$. 
Assume that for \( i \in [m] \), the weight \( w_i \), obtained from the cosine similarity calculation, follows the distribution \( \mathcal{N}(\mu, \sigma^2) \), where \( 0 < \mu < 1 \) and \( \sigma < 1 \).
\end{assumption}

\subsection{Proof of Theorem \ref{th:convergence}}
% 由于条件与IFCA相似，因此我们仅集中于gradient descent iteration. 不失一般性，我们专注于 $\theta_1$.
Since the conditions are similar to IFCA, we focus only on the gradient descent iteration. Without loss of generality, we concentrate on \( \theta_1 \).
\[
    \|\theta^+_1 - \theta^*_1\| = \|\theta_1 - \theta^*_1 - \frac{\gamma}{w}\sum_{i \in S_1}{w_i \nabla F_i(\theta_1)} \|,
\]
where \( F_i(\theta) := F_i(\theta; Z_i) \), with \( Z_i \) being the set of data points on the \( i \)-th client used to compute the gradient, and \( S_1 \) is the set of indices corresponding to the first cluster.
% 其中，\( F_i(\theta) := F_i(\theta; Z_i) \)，其中 \( Z_i \) 是第 \( i \) 个工作节点上用于计算梯度的数据点集合，而 \( S_1 \) 是第一个聚类对应的索引集合。
% 由于
Since
\[
    S_1 = (S_1 \cap S^*_1) \cup (S_1 \cap S^*_1),
\]
we have
% 所以我们有
\[
    \begin{aligned}
    \| \theta^+_1 - \theta^*_1 \| &= \|\underbrace{\theta_1 - \theta^*_1 - \frac{\gamma}{w}\sum_{i \in S_1 \cap S^*_1}{w_i \nabla F_i(\theta_1)}}_{T_1} \\
    &-  \underbrace{\frac{\gamma}{w}\sum_{i \in S_1 \cap \overline{S^*_1}}{\nabla F_i(\theta_1)}}_{T_2}\|.
    \end{aligned}
\]
% 根据三角不等式，我们有
According to triangle inequality, we have
\[
    \|\theta^+_1 - \theta^*_1\| \le \|T_1\| + \|T_2\|.
\]
% 我们定义 $\hat{\gamma}$ 如下
We define $\hat{\gamma}$ as
\[
    \hat{\gamma} = \frac{\gamma}{w} |S_1 \cap S^*_1|.
\]
% 我们首先关注 $\|T_1\|$ 的边界，将其分割为如下形式
We first focus on the bound of \( \|T_1\| \) and decompose it into the following form.
\[
    \begin{aligned}
    T_1 &= \underbrace{\theta_1 - \theta^*_1 - \hat{\gamma} \mu \nabla F^1(\theta_1)}_{T_{11}} \\
    &+ \hat{\gamma}(\underbrace{\mu \nabla F^1(\theta_1) - \frac{1}{|S_1 \cap S^*_1|}\sum_{i \in S_1 \cap S^*_1}{w_i \nabla F_i(\theta_1)}}_{T_{12}}).
    \end{aligned}
\]
% 根据函数的强凸性，我们有
Based on the strong convexity of the function, we have
\[
    \|T_{11}\| = \|\theta_1 - \theta^*_1 - \hat{\gamma} \mu \nabla F^1(\theta_1)\| \le (1-\frac{\hat{\gamma}\mu\lambda L}{\lambda + L})\|\theta_1 - \theta^*_1\|.
\]
% 另一方面，我们可以计算得到 $\|T_{12}\|^2$ 的期望
On the other hand, we can compute the expectation of \( \|T_{12}\|^2 \) as
\[
    E[\|T_{12}\|^2] = \frac{v^2\sigma^2 + \sigma^2 u^2 + v^2\mu^2}{n'|S_1 \cap S^*_1|}.
\]
% 又因为
Since $u^2 \le 1,\; \mu^2 \le 1,\; \sigma^2 \le 1$,
% 所以
we have
\[
    E[\|T_{12}\|] \le \frac{\sqrt{v^2\sigma^2 + \sigma^2 u^2 + v^2\mu^2}}{\sqrt{n'|S_1 \cap S^*_1|}} \le \frac{\sqrt{2v^2 + 1}}{\sqrt{n'|S_1 \cap S^*_1|}}.
\]
By Markov's inequality, for any $\delta_0 > 0$, with probability at least $1 - \delta_0$, we have

\begin{equation}
\label{eq:t_12}
    \|T_{12}\| \le \frac{\sqrt{2v^2 + 1}}{\delta_0\sqrt{n'|S_1 \cap S^*_1|}}.
\end{equation}
% 根据IFCA中的推导，我们可以得到
Based on the derivations in IFCA, we can obtain
\[
    |S_1 \cap S^*_1| \ge \frac{1}{4}p_1m.
\]
% 所以有
So we have
\[
    \hat{\gamma} \ge \frac{pm}{4Lw}.
\]
% 此外，我们可以认为 $w = \mu m$，如果 $m$ 的数量足够大。综上，我们可以得到
In addition, we can assume that \( w = \mu m \) if the number of \( m \) is large enough. Therefore, we can obtain
\begin{equation}
\label{eq:t_11}
    \|T_{11}\| \le (1-\frac{pm\mu\lambda}{8wL})\|\theta_1 - \theta^*_1\| = (1-\frac{p\lambda}{8L})\|\theta_1 - \theta^*_1\|.
\end{equation}
% 根据 Eq.\eqref{t_12} 和 Eq.\eqref{t_11} 我们可以得到
Based on Eq. \eqref{eq:t_12} and Eq. \eqref{eq:t_11}, we can obtain
\begin{equation}
\label{eq:t1}
     \|T_{1}\| \le (1-\frac{p\lambda}{8L})\|\theta_1 - \theta^*_1\| + \frac{2\sqrt{2v^2 + 1}}{\delta_0L\sqrt{pmn'}}.
\end{equation}

% 接下来我们关注 $\|T_2\|$ 的边界，首先我们定义
Next, we focus on the bound of \( \|T_2\| \). First, we define
\[
    T_{2j} := \sum_{S_1 \cap S^*_j}{w_i \nabla F_i(\theta_i)}.
\]
% $T_2$ 可以表示为
And \( T_2 \) can be expressed as $T_2 = \frac{\gamma}{w}\sum^k_{j=2}T_{2j}$.

% 结合 IFCA 的推导，我们可以得到
Combining the derivations from IFCA, we can obtain
\begin{equation*}
    \begin{aligned}
        T_{2j} =& |S_1 \cap S^*_j| \mu \nabla F^j(\theta_i) \\
        &+ \sum_{S_1 \cap S^*_j}{(\mu \nabla F^j(\theta_i) - w_i \nabla F_i(\theta_i))},
    \end{aligned}
\end{equation*}
\begin{equation*}
    \begin{aligned}
       E(\|\sum_{S_1 \cap S^*_j}{\mu \nabla F^j(\theta_i) - w_i \nabla F_i(\theta_i)}\|^2) \\
       \le \frac{\sqrt{(2v^2 + 1)|S_1 \cap S^*_j|}}{\sqrt{n'}}.
    \end{aligned}
\end{equation*}

With probability at least $1-\delta_1$, we have

\begin{equation}
\label{eq:t_2j}
    \|T_{2j}\| \le c_1|S_1 \cap S^*_j| +\frac{\sqrt{(2v^2 + 1)|S_1 \cap S^*_j|}}{\delta_1\sqrt{n'}}.
\end{equation}

With probability at least $1 - k\delta_1$, apply Eq.~\eqref{eq:t_2j} to all $j \ne 1$. Then, we have with probability at least $1 - k\delta_1$,

\begin{equation}
\label{eq:t2s}
     \|T_{2}\| \le \frac{c_1\gamma}{w}|S_1 \cap \overline{S^*_1}| + \frac{\gamma\sqrt{(2kv^2 + k)|S_1 \cap \overline{S^*_1}|}}{\delta_1w\sqrt{n'}}.
\end{equation}

% 根据 IFCA 的推导，我们可以得到
Based on the derivations in IFCA, we can obtain
\[
    E(|S_1 \cap \overline{S^*_1}|) \le c_1\frac{\eta^2m}{\alpha^2\lambda^2\Delta^4n'}.
\]
% 结合 Markov's inequality，我们有
Combining Markov's inequality, we have
\begin{equation}
\label{eq:s}
    |S_1 \cap \overline{S^*_1}| \le c_1\frac{\eta^2m}{\delta_2\alpha^2\lambda^2\Delta^4n'}.
\end{equation}

% 结合 Eq.\eqref{eq:t2s} 和 Eq.\eqref{eq:s} 我们可以得到
Combining Eq.\eqref{eq:t2s} and Eq.\eqref{eq:s}, we can obtain
\begin{equation}
\label{eq:t2}
    \|T_2\| \le  c_1\frac{\eta^2m}{\delta_2 \alpha^2\lambda^2\Delta^4 wn'} + c_2\frac{\eta\sqrt{2kmv^2 + km}}{\delta_1\sqrt{\delta_2}\alpha\lambda L\Delta^2wn'}.
\end{equation}

% 结合 Eq.\eqref{eq:t1} 和 Eq.\eqref{eq:t2} 我们可以得到
Combining Eq.\eqref{eq:t1} and Eq.\eqref{eq:t2}, we can obtain
\[
    \begin{aligned}
    \|\theta^+_1 - \theta^*_1\| \le& (1-\frac{p\lambda}{8L})\|\theta_1 - \theta^*_1\| + \frac{2\sqrt{2v^2 + 1}}{\delta_0L\sqrt{pmn'}} \\
    &+  c_1\frac{\eta^2m}{\delta_2 \alpha^2\lambda^2\Delta^4 wn'} + c_2\frac{\eta\sqrt{2kmv^2 + km}}{\delta_1\sqrt{\delta_2}\alpha\lambda L\Delta^2wn'}.
    \end{aligned}
\]

% We choose $\delta_0 = \frac{\delta}{4}$, $\delta_1 = \frac{\delta}{4k}$, and $\delta_2 = \frac{\delta}{4}$, at least $1 - \delta$, we have
Let \( \delta \geq \delta_0 + k\delta_1 + \delta_2 + 2\exp(-cpm) \) be the failure probability of this iteration, and choose \( \delta_0 = \frac{\delta}{4}, \delta_1 = \frac{\delta}{4k}, \delta_2 = \frac{\delta}{4} \).  
Then the failure probability is upper bounded by \( \delta \) as long as \( 2\exp(-cpm) \leq \frac{\delta}{4} \), which is guaranteed by our assumption that \( p \gtrsim \frac{1}{m} \log(mn') \).  
Therefore, we conclude that with probability at least \( 1 - \delta \),
\[
    \begin{aligned}
    \|\theta^+_1 - \theta^*_1\| \le& (1-\frac{p\lambda}{8L})\|\theta_1 - \theta^*_1\| + \frac{c_0\sqrt{2v^2 + 1}}{\delta L\sqrt{pmn'}} \\
    &+  c_1\frac{\eta^2m}{\delta \alpha^2\lambda^2\Delta^4 wn'} + c_2\frac{\eta k\sqrt{2kmv^2 + km}}{\delta^{\frac{3}{2}}\alpha\lambda L\Delta^2wn'},
    \end{aligned}
\]
which completes the proof.
\end{document}